\UseRawInputEncoding
\documentclass[12pt, draftclsnofoot, onecolumn]{IEEEtran}
\usepackage{amsmath,amssymb,amsfonts}
\usepackage{cite}
\usepackage{url}
\usepackage{xcolor}
\usepackage{graphicx}
\usepackage{subfigure}
\usepackage{fancyhdr}
\usepackage{mdwmath}
\usepackage{mdwtab}
\usepackage{caption}
\usepackage{amsthm}
\usepackage{algorithmic}
\usepackage{textcomp}
\usepackage{makecell,multirow}
\usepackage{lipsum}
\usepackage{ccmap}
\usepackage{epstopdf}
\newtheorem{remark}{Remark}
\newtheorem{theorem}{Theorem}
\newtheorem{lemma}{Lemma}
\newtheorem{corollary}{Corollary}

\makeatletter
\def\ScaleIfNeeded{%
\ifdim\Gin@nat@width>\linewidth \linewidth \else \Gin@nat@width
\fi } \makeatother
\begin{document}
\title{Performance Analysis and Power Allocation of Joint Communication and Sensing Towards Future Communication Networks}
\author{
Meng Liu,~\IEEEmembership{Student Member,~IEEE,}
Minglei Yang,~\IEEEmembership{Member,~IEEE,}
        Huifang Li,
        Kun Zeng,~\IEEEmembership{Member,~IEEE,}
        Zhaoming Zhang,~\IEEEmembership{Student Member,~IEEE,}
        Xiancheng Cheng,~\IEEEmembership{Student Member,~IEEE,}
        Arumugam Nallanathan,~\IEEEmembership{Fellow,~IEEE,}\\
        Derrick Wing Kwan Ng,~\IEEEmembership{Fellow,~IEEE,}
        Guangjian Wang,~\IEEEmembership{Member,~IEEE}
\thanks{M. Liu, M. Yang, Z. Zhang, and X. Cheng are with the National Laboratory of Radar Signal Processing, Xidian University, Xi'an, Shaanxi 710071, China (e-mail: liumeng2021@163.com; mlyang@xidian.edu.cn; zmzhang\_sx@163.com; xcheng@syr.edu).}
\thanks{H. Li is with State Key Laboratory of Integrated Services Networks, Xidian University, Xi'an, Shaanxi 710071, China (e-mail:huifanglixd@gmail.com)}
\thanks{K. Zeng and G. Wang are with Huawei Technologies Co., Ltd, Chengdu, Sichuan 611731 China (e-mail: kun.zeng@huawei.com; wangguangjian@huawei.com)}
\thanks{A. Nallanathan is with the School of Electronic Engineering and Computer Science, Queen Mary University of London, London E1 4NS, U.K. (e-mail: a.nallanathan@qmul.ac.uk).}
\thanks{D. W. K. Ng is with the School of Electrical Engineering and Telecommunications, University of New South Wales, Kensington, NSW 2052, Australia (e-mail: w.k.ng@unsw.edu.au).}
}
\maketitle
\vspace{-20 mm}
\begin{abstract}
\vspace{-5 mm}
To mitigate the radar and communication frequency overlapping caused by massive devices access, we propose a novel joint communication and sensing (JCS) system in this paper, where a micro base station (MiBS) can realize target sensing and cooperative communication simultaneously. Concretely, the MiBS, as the sensing equipment, can also serve as a full-duplex (FD) decode-and-forward (DF) relay to assist the end-to-end communication. To further improve the spectrum utilization, non-orthogonal multiple access (NOMA) is adopted such that the communication between the macro base station (MaBS) and the Internet-of-Things (IoT) devices. To facilitate the performance evaluation, the exact and asymptotic outage probabilities, ergodic rates, sensing probability of the system are characterized. Subsequently, two optimal power allocation (OPA) problems of maximizing the received signal-to-interference-plus-noise ratio of sensing signal and maximizing the sum rate for communication are designed that are solved by means of the Lagrangian method and function monotonicity. The simulation results demonstrate that: 1) the proposed JCS NOMA system can accomplish both communication enhancement and sensing function under the premise of the same power consumption as non-cooperative NOMA; 2) the proposed OPA schemes manifest superiorities over a random power allocation scheme.
\end{abstract}
\vspace{-5 mm}
\begin{IEEEkeywords}
\vspace{-5 mm}
Joint communication and sensing, cooperative communication, non-orthogonal multiple access, full-duplex.
\end{IEEEkeywords}
\vspace{-6 mm}
\section{Introduction}
\vspace{-4 mm}
\subsection{Background}
\vspace{-2 mm}
\IEEEPARstart{W}{ith} the development of the Internet-of-Things (IoT), a large number of devices are connected to mobile communication networks explosively, which leads to the proximity of communication and radar frequency bands directly \cite{Liuf20,Hongh21,Hassanien19}. In order to alleviate the spectrum congestion, various methods have been proposed, among which joint communication and radar (JCR) \cite{4} and joint communication and sensing (JCS) \cite{Wild21} are considered to be more feasible. In the past few decades, two feasible designs based on JCR systems have been actively explored, one is to design a shared waveform suitable for the simultaneous operation of radar and communication, the other is to design a pragmatic scheme that can enable the co-existence of radar and communication \cite{wang20c}. Nevertheless, the shared waveform design for JCR systems requires advanced radio frequency (RF) port processing, which increases the potential hardware cost, while the co-existence scheme would cause cross interference between the radar and communication operations \cite{Zhengle19}.

With the completion of the fifth-generation (5G) new radio (NR) deployment, the reliability, delay, and user experience have reached unprecedented accomplishments. Recently, looking forward to the future multi-functional network application scenarios such as smart city, smart transportation, and smart environment, applying mobile communications to data transmission only has been shown insufficient to cater to the requirements \cite{chenna21}. A prevalent prospect has been mapped out in the industry and academia to realize JCS \cite{Pin20}. Indeed, heterogeneous wireless applications are omnipresent and the IoT has been widely used in various domains, which lays a solid foundation for the implementation of JCS \cite{Chettri20}. Specifically, the entire network is used for smart transportation, smart home environment, smart healthcare, airport clearance of sensitive areas \cite{Saad20,Chowdhury20,Tataria21}. Indeed, the emergency needs have fueled the fusion between communication and sensing and transform the IoT to the Internet-of-Intelligent-Things (IoIT) \cite{Tariq20}. For instance, the authors in \cite{Zhangqx20} proposed a JCS system, which could reduce the delay of information sharing among vehicles and improve the date rate, while guaranteeing the safety of autonomous driving. Besides, in \cite{Chenx20}, a JCS system of unmanned aerial vehicle (UAV) based on the beam sharing was studied, in which UAVs could accomplish the functions of JCS. All the mentioned articles show that JCS is an effective way that can provide convenience for our lives.

The implementation of JCS requires massive end-to-end data transmission. However, it is tricky to realize information transmission with limited spectral resources. To address this issue, non-orthogonal multiple access (NOMA), as a core technology to resolve the potential of the future communication networks, can deliver significant improvements in delay, energy efficiency, supported numbers of users, and throughput in contrast to the conventional orthogonal multiple access (OMA). Specifically, OMA can only serve a single user by exploiting the resources of time/frequency/code domain orthogonally, while NOMA can provide a balanced service comparatively for serving multiple devices simultaneously while guaranteeing the required quality-of-service (QoS) of devices \cite{liuy17,17,you21}. The principle of NOMA is to allow controlled multiple access interference and adopts the successive interference cancellation (SIC) to decode the desired signals at the receiver \cite{dingzhb}. For instance, in \cite{Sun20}, the traffic offloading scheme of a device-to-device (D2D) network was investigated by using NOMA and the resources were optimized such that the system capacity was maximized. Also, the multi-constraint optimization algorithm based on a neoteric gravitational search was proposed in \cite{huifang21} to maximize the throughput in wireless cached NOMA UAV systems. The aforementioned works manifest that NOMA yields significant performance gains over OMA in various communication scenarios, serving as a promising candidate for enabling future mobile communication systems.

In order to further enhance the system robustness and to enlarge coverage area, cooperative relaying has been introduced into NOMA systems \cite{24}. For example, the authors employed the large intelligent reflecting surface as a relay to assist the end-to-end transmission of the NOMA system in \cite{Bariah21}, and the pairwise error probability expressions of users were derived to evaluate the system performance. Furthermore, in \cite{Xingwang20}, the outage probabilities and ergodic rates of users were analyzed in the cooperative and non-cooperative NOMA scenarios, which demonstrated that cooperative NOMA could enhance the system performance. In addition, a novel cooperative NOMA system was explored in \cite{zhang17}, where the near user acted as a relay to assist a far user to communicate with the source and the closed-form expressions of the optimal power allocation (OPA) were obtained to maximize the sum rate with low computational complexity.

Different from the half-duplex (HD) relaying adopted in \cite{24,Bariah21,Xingwang20,zhang17}, which needs two orthogonal time slots for signal transmission, the full-duplex (FD) relaying requires only one time slot that further improves the system spectral efficiency. Nevertheless, loop self-interference (LSI) is unavoidable in the FD systems. To this end, the authors in \cite{Riihonen11} adopted natural isolation, time-domain elimination, and spatial suppression such that the effects of LSI could be reduced in the FD systems. Furthermore, the authors in \cite{8666065} studied the FD-NOMA vehicle-to-everything system and demonstrated that FD could achieve a positive improvement in latency, while LSI caused a significant negative impact. Moreover, in a multiple downlink and uplink users FD-NOMA system \cite{Sunya17}, the authors unveiled that FD mode could improve the system throughput and enable a fairer resource allocation compared with HD mode.

\vspace{-5 mm}
\subsection{Motivation and Contribution}
\vspace{-2 mm}
Previous works laid solid theoretical foundations for establishing JCS systems. However, the performance gain of the dual-function JCS system has not been fully exploited yet in practice. In order to break through the dilemma, there are serval attempts in the literature for investigating JCR systems such that the communication and sensing can be realized respectively. In \cite{Chiriyath16}, to synchronize communication and sensing, the authors considered that radar can fulfill two functions simultaneously, that is, target tracking and cooperative communication. In particular, the inner bound of SIC and the Fisher information were discussed to evaluate the system performance. Ulteriorly, the estimation rate of radar and the performance bounds of the radar communication coexistence system were further developed to verify the system performance in \cite{Chiriyath17}. While the detailed features of the relay was not discussed exhaustively in \cite{Chiriyath16} and \cite{Chiriyath17}. The structures of single-carrier and multi-carrier JCR systems were devised in \cite{Wangfan19} and \cite{wangfa19} respectively, and the OPA schemes were analyzed to maximize the communication sum rate or the radar received signal-to-interference-plus-noise ratio (SINR). Moreover, in \cite{Mehmet21}, a downlink NOMA transceiver was designed from the perspective of radar and communication waveforms coexistence. The advancement of the proposed scheme was demonstrated in terms of radar ambiguity and bit-error rate. Also, it was pointed out that waveform coexistence could be applied to the JCS systems. However, in these papers, the authors used radar for sensing, which is difficult to implement. In practice, radar and communication are built independently with separated operating systems, and the aliasing of radar and communication signals generates obstacles to the signal processors. Moreover, the likelihood that radar provides additional assistance for communications is faint, especially in the military field.

As a remedy, the authors in \cite{xuwen18} studied a JCS cognitive OFDM-NOMA system in which the cognitive source sensed the state of the primary user to access the primary network frequency band opportunistically. In particular, an OPA algorithm was designed to maximize the system capacity under the premise of tolerable interference caused by the cognitive source. Also, in the cognitive Industrial IoT system \cite{liuxin20tii}, multiple sensing stations worked together and transmitted the sensing information to the data center by non-orthogonal communication to improve the data transmission efficiency. Besides, the JCS NOMA network framework was proposed recently in \cite{Qiaoi21}, in which the transmitted and received beamforming were designed for minimizing the computational error and maximizing weighted sum rate. Indeed, \cite{Mehmet21,xuwen18,liuxin20tii,Qiaoi21} provided embryonic forms for the development of JCS NOMA systems, whereas the JCS cooperative scheme was not considered, and the performance analysis of the system was lacked as well. From a practical point of view, the integration of JCS and NOMA can increase spectrum utilization enormously, improve communication reliability, and ensure resource allocation fairness. Therefore, the analysis and optimization of JCS NOMA networks are of tremendous significance.

Against the background, in this paper, we consider a JCS NOMA system where a micro base station (MiBS) can realize the functions of the cooperative relaying and target sensing, simultaneously. The contributions of this paper are summarized as follows:
\begin{itemize}
\item For the proposed JCS FD-NOMA system, the exact and asymptotic expressions of the outage probability for the IoT devices are obtained as well as the diversity orders in the high signal-to-noise ratio (SNR) regime. Furthermore, the ergodic rates of the IoT devices are also investigated. Due to the particularity of the expressions, we wield the integral and approximate expressions instead. The results indicate that on the basis of realizing the extra sensing function, the ergodic rate of the proposed JCS FD-NOMA system outperforms the traditional non-cooperative NOMA system without requiring additional power consumption.
\item The sensing probability of the MiBS is derived to evaluate the sensing performance of the JCS FD-NOMA system. The obtained results show that the MiBS sensing behavior is sensitive to the transmitted power of the MaBS.
\item To further improve the performance of the considered system, we formulate two OPA problems: \emph{i)} sensing-centric design (SCD), i.e., maximizing the received SINR of the sensing signal at the MiBS, while the IoT devices satisfy the required QoS constraint; \emph{ii)} communication-centric design (CCD), i.e., maximizing the sum rate of the IoT devices, while the MiBS meets the minimum SINR requirement. Finally, we derive the closed-form expressions of the OPA coefficients analytically to improve the performance of the considered system.
\end{itemize}

\vspace{-5 mm}
\subsection{Organization and Notations}
\vspace{-2 mm}

The remainder of this paper is organized as follows. Section II discusses the system model of the JCS FD-NOMA system. In Section III, the performance of the considered system is evaluated by deriving the exact and asymptotic outage probabilities and the corresponding diversity orders. The ergodic rates of the IoT devices are derived in Section IV. Section V analyzes the sensing performance of the system. Two OPA problems are formulated in Section VI. The numerical results are provided to demonstrate the correctness of our theoretical analysis in Section VII before the conclusion in Section VIII.

In this paper, we use $\mathbb{E}\left[  \cdot  \right]$ to represent the expectation operation. ${\rm{{\cal C}{\cal N}}}\left( {\mu ,{\sigma ^2}} \right)$ symbolizes a complex valued Gaussian random variable with mean $\mu$ and variance ${\sigma ^2}$. $ \triangleq $ denotes the definition operation. $\text{Ei}\left(  \cdot  \right)$ is the exponential integral function. $\mathbb{C}$ represents the set of complex numbers. ${f_X}\left(  \cdot  \right)$ and ${F_X}\left(  \cdot  \right)$ are the probability density function (PDF) and cumulative distribution function (CDF) of the random variable $X$, respectively.

\vspace{-2 mm}
\section{System Model}
\vspace{-1 mm}
\begin{figure} [!t]
\vspace{-1.2 cm}  
\setlength{\abovecaptionskip}{-2 pt}  
\setlength{\belowcaptionskip}{-2 pt }
\centering
\includegraphics[width=8cm]{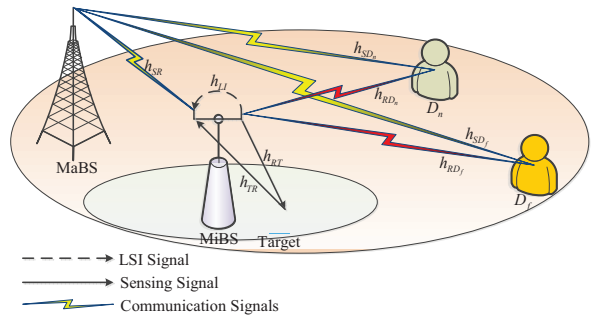}
 \caption{An illustration of the dual-function MiBS cooperative NOMA network.}
 \vspace{-0.5 cm}   
\end{figure}
We consider a JCS network as shown in Fig. 1, where the MaBS $S$ is mainly in charge of communication with IoT devices, while the MiBS $R$ is responsible for sensing surrounding unknown targets and to facilitate cooperative communication between the MaBS and two IoT devices, namely the far IoT device $D_f$ and the near IoT device $D_n$. Specifically, in order to improve the communication performance of the system, the MiBS can also act as a FD decode-and-forward (DF) relay to assist the end-to-end communication while carrying out sensing a point target $T$. We further assume that the MiBS is equipped with two antennas, one for transmitting signals and the other for receiving signals, while other equipment is equipped with a single-antenna. For the convenience of the presentation, the channel coefficients of $S \to R$, $S \to D_f$, $S\to D_n$, $R\to T$, $T\to R$, $R\to D_f$, and $R\to D_n$ are denoted by ${ h_i}$, $i \in \left\{ {SR,S{D_f},S{D_n},RT,TR,R{D_f},R{D_n}} \right\}$, respectively.\footnote{We further assume that the perfect channel state information (CSI) between the nodes can be obtained by using some existing channel training methods, e.g., \cite{lvlu19}.} It is further assumed that all the links are subject to Rayleigh fading, i.e., the channel gain ${ \rho _i} = {\left| {{ h_i}} \right|^2}$ all obeys exponential distribution.\footnote{It can be obtained that the PDF and CDF of ${ \rho _i}$ are, respectively, given by ${f_{{{ \rho }_i}}}\left( x \right) = {{\exp \left( { - {x \mathord{\left/
 {\vphantom {x {{\beta _i}}}} \right.
 \kern-\nulldelimiterspace} {{\beta _i}}}} \right)} \mathord{\left/
 {\vphantom {{\exp \left( { - {x \mathord{\left/
 {\vphantom {x {{\beta _i}}}} \right.
 \kern-\nulldelimiterspace} {{\beta _i}}}} \right)} {{\beta _i}}}} \right.
 \kern-\nulldelimiterspace} {{\beta _i}}}$ and ${F_{{{ \rho }_i}}}\left( x \right) = 1 - \exp \left( { - {x \mathord{\left/
 {\vphantom {x {{\beta _i}}}} \right.
 \kern-\nulldelimiterspace} {{\beta _i}}}} \right)$, where ${\beta _i} = {h \mathord{\left/
 {\vphantom {h {\sqrt {1 + d_i^\alpha } }}} \right.
 \kern-\nulldelimiterspace} {\sqrt {1 + d_i^\alpha } }}$ denotes the channel variance, $h \sim {\cal{C}\cal{N}}\left( {0,\Omega } \right)$ is the complex channel coefficient, $\alpha$ denotes the path loss exponent, and $d_i$ is the distance between the nodes \cite{8302918}.} Next, we will discuss the signal transmission process for the proposed cooperative NOMA scheme.

For the non-cooperative links, during the $t$-th time slot, $S$ intends to send signal ${y_c} = \sqrt {{a_n}{P_{com}}} {x_n}\left( t \right) + \sqrt {{a_f}{P_{com}}} {x_f}\left( t \right)$ to $D_f$ and $D_n$ concurrently, where $P_{com}$ denotes the transmit power of $S$, ${x_f}\left(t\right) \in \mathbb{C}$ and ${x_n}\left(t\right)\in \mathbb{C}$ are the desired signals for $D_f$ and $D_n$ with $\mathbb{E}\left[ {{{\left| {{x_f}} \right|}^2}} \right] = \mathbb{E}\left[ {{{\left| {{x_n}} \right|}^2}} \right] = 1$, $a_f$ and $a_n$ represent the power allocation coefficients such that $a_n+a_f=1$ and $a_f>a_n$, respectively. Thus, the received signals at $D_f$ and $D_n$ can be expressed as
\vspace{-1 mm}
\begin{equation}
{y_{{i_1}}} = {h_{{i_1}}}\left( {\sqrt {{a_n}{P_{com}}} {x_n}\left( t \right) + \sqrt {{a_f}{P_{com}}} {x_f}\left( t \right)} \right) + {n_{{i_1}}},
\end{equation}
\!where ${i_1} \in \left\{ {S{D_f},S{D_n}} \right\}$, ${n_{i_1}} \sim {\rm{{\cal C}{\cal N}}}\left( {0,{N_0}} \right)$ denotes the additive white Gaussian noise (AWGN) and $N_0$ is the noise power.

For the cooperative links, $R$ transmits ${y_r} = \sqrt {{P_{sen}}} {x_r}\left( t -\tau _1\right)$ to sense the target, then receives the signals from $S$ and the target echo, where ${P_{sen}}$ is the power budget of $R$, ${x_r}\in \mathbb{C}$ denotes the transmitted signal of the MaBS with $\mathbb E\left[ {{{\left| {{x_r}} \right|}^2}} \right] = 1$, and $\tau_1$ is the processing delay of $x_r$ at $R$. Therefore, the received signal at $R$ can be expressed as
\vspace{-1 mm}
\begin{equation}\label{ysr}
{y_{SR}} \!=\! {h_{SR}}\!\left(\! {\sqrt {{a_n}{P_{com}}} {x_n}\!\left( t \right) \!+ \!\sqrt {{a_f}{P_{com}}} {x_f}\!\left( t \right)} \!\right) + {h_{RR}}\sqrt {\delta {P_{sen}}} {x_r}\!\left( {t \!- \!{\tau _1}} \right) + {h_{LI}}\sqrt {\omega {P_{sen}}} {x_{LI}}\!\left( {t\! - \!{\tau _2}} \right) + {n_{SR}},
\end{equation}
\!where ${n_{SR}} \sim {\rm{{\cal C}{\cal N}}}\left( {0,{N_0}} \right)$ is the AWGN at $R$, $x_{LI}\in \mathbb{C}$ denotes the LSI signal, $h_{LI}$ is the channel coefficient of LSI channel at $R$, $\omega  \in \left\{ {0,1} \right\}$ denotes the switching operation factor of $R$,\footnote{It should be pointed that when $\omega=0$, $R$ is in half duplex (HD) mode, and when $\omega=1$, $R$ is in FD mode. } $ \tau_2 $ represents the processing delay of $x_{LI}$ at $R$ with $\tau_2<\tau_1$, and $h_{RR}$ is the new channel state coefficient which captures the joint impacts carried by $h_{RT}$ and $h_{TR}$ and $\delta  \in \left[ {0,1} \right]$\footnote{When $\delta =0$, it indicates that no echo exists in the received signal of the MiBS, and there is no target in the sensing range.} denotes the power reflection factor caused by the target \cite{Vincentl13}. In the following, we introduce the CDF and PDF of $\rho_{RR}$.
\vspace{-3 mm}
\begin{lemma}
Assuming that the target is moving slowly, we consider that the channel environment experienced by the signals of $R\to T$ and $T \to R$ remains unchanged approximately, i.e., ${\rho_{RT}} = {\rho_{TR}}$. Then the CDF and PDF of $\rho _{RR}$ are expressed as
\vspace{-1 mm}
\begin{align}\label{rhorr}
{F_{{\rho _{RR}}}}\left( x \right)& = 1 - {e^{ - \frac{{\sqrt x }}{{{\beta _{RR}}}}}},\\\label{rhorrpdf}
{f_{{\rho _{RR}}}}\left( x \right) &= \frac{1}{{2{\beta _{RR}}\sqrt x }}{e^{ - \frac{{\sqrt x }}{{{\beta _{RR}}}}}}.
\end{align}
\end{lemma}

\begin{proof}
By using the CDF of $\rho_i$, we can obtain
\begin{align}\nonumber
 {F_{{\rho _{RR}}}}\left( x \right)& = \Pr \left( {{\rho _{RR}} = {\rho _{RT}}{\rho _{TR}} \leqslant x} \right) \hfill \\\nonumber
  & = \int_0^{\sqrt x } {{f_{{\rho _{RT}}}}\left( t \right)} dt = \int_0^{\sqrt x } {{f_{{\rho _{TR}}}}\left( t \right)} dt \hfill \\\label{rhor}
   &= 1 - {e^{ - \frac{{\sqrt x }}{{{\beta _{RT}}}}}} = 1 - {e^{ - \frac{{\sqrt x }}{{{\beta _{TR}}}}}} .
\end{align}

Using the basic variable substitution, (\ref{rhorr}) can be obtained, and then taking the derivative of ${F_{{\rho _{RR}}}}\left( x \right)$ with respect to $x$, (\ref{rhorrpdf}) can be obtained.
\end{proof}

After receiving the signals from $S$ and the echo, the MiBS, serving as a FD DF relay for communication, decodes and forwards the received signals to $D_f$ and $D_n$. The received signals at the IoT devices can be expressed as\footnote{It should be emphasized that if $R$ can decode the communication signals successfully, then it will forward them to the IoT devices, otherwise $R$ remains silent and the cooperative links will be in a state of outage.}
\begin{equation}\label{yrd}
{y_{{i_2}}} = {h_{{i_2}}}\left( {\sqrt {{a_n}{P_{sen}}} {x_n}\left( {t - \tau } \right) + \sqrt {{a_f}{P_{sen}}} {x_f}\left( {t - \tau } \right)} \right) + {n_{{i_2}}},
\end{equation}
where ${i_2} \in \left\{ {R{D_f},R{D_n}} \right\}$, ${n_{{i_2}}} \sim {\rm{\cal{CN}}} \left( {0,{N_0}} \right)$ denotes the AWGN, and $\tau=\tau_1+\tau_2$.

The received SINRs of the communication signals can be divided into non-cooperative and cooperative links as well.
\vspace{-2 mm}
\subsection{Non-cooperative Links}
\vspace{-1 mm}

For the non-cooperative links, the MaBS can communicate with the IoT devices directly. According to the NOMA protocol, $x_f$ is first decoded at $D_n$, and then $x_n$ is decoded by the application of SIC. Thus, the SINRs of $D_n$ decoding $x_f$ and $x_n$ are respectively expressed as
\vspace{-2 mm}
\begin{align}\label{sinrsdnf}
\gamma _{S{D_n}}^{{D_f}} &= \frac{{{a_f}{\rho _{S{D_n}}}{\gamma _c}}}{{{a_n}{\rho _{S{D_n}}}{\gamma _c} + 1}},\\\label{sinrsdn}
\gamma _{S{D_n}}^{{D_n}} &= {a_n}{\rho _{S{D_n}}}{\gamma _c},
\end{align}
where ${\gamma _c} = {{{P_{com}}} \mathord{\left/
{\vphantom {{{P_{com}}} {{N_0}}}} \right.
\kern-\nulldelimiterspace} {{N_0}}}$ denotes the transmit SNR of the MaBS.

The received SINR of $D_f$ decoding its own signal over the direct link can be expressed as
\vspace{-1 mm}
\begin{equation}
\label{sinrsdf}
\gamma _{S{D_f}}^{{D_f}} = \frac{{{a_f}{\rho _{S{D_f}}}{\gamma _c}}}{{{a_n}{\rho _{S{D_f}}}{\gamma _c} + 1}}.
\end{equation}

\vspace{-4 mm}
\subsection{Cooperative Links}
\vspace{-1 mm}
For the cooperative links, $R$ decodes the IoT devices' desired signals in turn by employing SIC and treats the echo and LSI as interferences for communication. Thus, the received SINRs of $x_f$ and $x_n$ at $R$ can be expressed as
\begin{align}\label{sinryrf}
\gamma _{SR}^{{D_f}} &= \frac{{{a_f}{\rho _{SR}}{\gamma _c}}}{{{a_n}{\rho _{SR}}{\gamma _c} + {\rho _{RR}}\delta {\gamma _r} + {\rho _{LI}}\omega {\gamma _r} + 1}},\\\label{sinryrn}
\gamma _{SR}^{{D_n}}& = \frac{{{a_n}{\rho _{SR}}{\gamma _c}}}{{{\rho _{RR}}\delta {\gamma _r} + {\rho _{LI}}\omega {\gamma _r} + 1}},
\end{align}
respectively, where ${\gamma _r} = {{{P_{sen}}} \mathord{\left/
 {\vphantom {{{P_{rad}}} {{N_0}}}} \right.
 \kern-\nulldelimiterspace} {{N_0}}}$ denotes the transmit SNR of the MiBS.

If $x_f$ is decoded successfully at $R$, i.e., the received SINR of $x_f$ exceed the outage threshold, $x_f$ will be sent to $D_f$. Then, the SINR of $x_f$ at $D_f$ can be expressed as
\begin{equation}
\label{sinrff}
\gamma _{R{D_f}}^{{D_f}} = \frac{{{a_f}{\rho _{R{D_f}}}{\gamma _r}}}{{{a_n}{\rho _{R{D_f}}}{\gamma _r} + 1}}.
\end{equation}

Similarly, if both $x_f$ and $x_n$ are decoded successfully at $R$, the signals will be sent to the $D_n$. At $D_n$, the high-power signal $x_f$ will be decoded first, and then cancel it until $x_n$ is detected. Therefore, the received SINRs of $x_f$ and $x_n$ at $D_n$ are respectively expressed as
\begin{align}\label{sinrdnf}
\gamma _{R{D_n}}^{{D_f}} &= \frac{{{a_f}{\rho _{R{D_n}}}{\gamma _r}}}{{{a_n}{\rho _{R{D_n}}}{\gamma _r} + 1}},\\\label{sinrdnn}
\gamma _{R{D_n}}^{{D_n}} &= {a_n}{\rho _{R{D_n}}}{\gamma _r}.
\end{align}

\section{Outage Probability Analysis}
\vspace{-1 mm}
In this section, we derive the exact and asymptotic outage probabilities of the IoT devices as well as the diversity orders to evaluate the communication performance of the proposed JCS NOMA system.

\vspace{-5 mm}
\subsection{Exact Outage Probability}
\vspace{-1 mm}
The outage event generates at $D_f$, when one of the following two events occurs. Case \emph{i)} $R$ cannot decode $x_f$ successfully and $D_f$ cannot decode $x_f$ over the direct link; Case \emph{ii)} $R$ decodes and forwards $x_f$ to $D_f$, while $D_f$ fails to decode the expected signal over the direct link or with the aid of the MiBS. Thus, the outage probability of $D_f$ over the direct and the MiBS cooperative links is expressed as\footnote{In this paper, the selection combining (SC) scheme is adopted to calculate the outage probabilities of the IoT devices, since it has low implementation cost and does not reduce the diversity gain. Our analysis can be extended to the maximal ratio combining (MRC) scheme, which will be set aside for our future work.}
\begin{equation}
\label{oprdf}
{\rm P}_{\rm{out}}^{{D_f}} = \underbrace {\Pr \left( {\max \left( {\gamma _{SR}^{{D_f}},\gamma _{S{D_f}}^{{D_f}}} \right) < {\gamma _{{\text{t}}{{\text{h}}_f}}}} \right)}_{{I_1}} + \underbrace {\Pr \left( {\gamma _{SR}^{{D_f}} \geqslant {\gamma _{{\text{t}}{{\text{h}}_f}}},\max \left( {\gamma _{R{D_f}}^{{D_f}},\gamma _{S{D_f}}^{{D_f}}} \right) < {\gamma _{{{\rm{th}}_f}}}} \right)}_{{I_2}},
\end{equation}
\!\!where ${\gamma _{{\text{t}}{{\text{h}}_f}}}$ denotes the SNR threshold of $x_f$, $I_1$ refers to Case \emph{i} and $I_2$ refers to Case \emph{ii}. The exact outage probability expression of $D_f$ for the considered JCS NOMA system is given in the following theorem.
\begin{theorem}
The exact outage probability of the far IoT device is expressed as
\begin{align}\nonumber
  {\rm P}_{\rm{out}}^{{D_f}} &= \left( {1 - {e^{ - \frac{{{\theta _1}}}{{{\beta _{S{D_f}}}}}}}} \right)\left\{ {1 - \frac{{{\beta _{SR}}}}{{2{\beta _{RR}}\left( {{\beta _{SR}} + {\beta _{LI}}\omega {\gamma _r}{\theta _1}} \right)}}\sqrt {\frac{{\pi {\beta _{SR}}}}{{\delta {\gamma _r}{\theta _1}}}} } \right. \hfill \\\label{oprf1}
   &\times \left. {{e^{\frac{{{\beta _{SR}}}}{{4\beta _{RR}^2\delta {\gamma _r}{\theta _1}}} - \frac{{{\theta _1}}}{{{\beta _{SR}}}} - \frac{{{\varphi _1}}}{{{\beta _{R{D_f}}}}}}}\left[ {1 - {\rm{erfc}}\left( {\frac{1}{{2{\beta _{RR}}}}\sqrt {\frac{{{\beta _{SR}}}}{{\delta {\gamma _r}{\theta _1}}}} } \right)} \right]} \right\},
\end{align}
\!\!where ${\theta _1} = {{{\gamma _{{{\rm{th}}_f}}}} \mathord{\left/
 {\vphantom {{{\gamma _{{{\rm{th}}_f}}}} {\left( {{a_f}{\gamma _c} - {a_n}{\gamma _c}{\gamma _{{{\rm{th}}_f}}}} \right)}}} \right.
 \kern-\nulldelimiterspace} {\left( {{a_f}{\gamma _c} - {a_n}{\gamma _c}{\gamma _{{{\rm{th}}_f}}}} \right)}}$, ${\varphi _1} = {{{\gamma _{{{\rm{th}}_f}}}} \mathord{\left/
 {\vphantom {{{\gamma _{{{\rm{th}}_f}}}} {\left( {{a_f}{\gamma _r} - {a_n}{\gamma _r}{\gamma _{{{\rm{th}}_f}}}} \right)}}} \right.
 \kern-\nulldelimiterspace} {\left( {{a_f}{\gamma _r} - {a_n}{\gamma _r}{\gamma _{{{\rm{th}}_f}}}} \right)}}$ with ${{a_f} > {a_n}{\gamma _{{{\rm{th}}_f}}}}$, otherwise ${\rm P}_{\rm{out}}^{{D_f}} = 1$ and $\rm{erfc}\left(  \cdot  \right)$ denotes the complementary error function which can be expressed as

\begin{align}
\label{erfc}
{\rm{erfc}} \left( z \right) = \frac{2}{{\sqrt \pi  }}\int_z^\infty  {{e^{ - {t^2}}}} dt.
\end{align}
\end{theorem}
\begin{proof}
See Appendix A.
\end{proof}

Similarly, the outage event generates at $D_n$ when: Case \emph{iii)} $R$ cannot decode $x_f$ or $x_n$ successfully and $D_n$ cannot decode $x_f$ or $x_n$ over the direct link; Case \emph{iv)} $R$ decodes and forwards the signals to $D_n$, while $D_n$ fails to decode the signals over the direct link or with the aid of the MiBS. Thus, the outage probability of $D_n$ over the direct and the MiBS cooperative links is expressed as
\begin{align}\nonumber
{\rm P}_{\rm{out}}^{{D_n}} &= \underbrace {\Pr \left( {\min \left( {\frac{{\gamma _{SR}^{{D_f}}}}{{{\gamma _{{{\rm{th}}_f}}}}},\frac{{\gamma _{SR}^{{D_n}}}}{{{\gamma _{{{\rm{th}}_n}}}}}} \right) < 1,\min \left( {\frac{{\gamma _{S{D_n}}^{{D_f}}}}{{{\gamma _{{{\rm{th}}_f}}}}},\frac{{\gamma _{S{D_n}}^{{D_n}}}}{{{\gamma _{{{\rm{th}}_n}}}}}} \right) < 1} \right)}_{{I_3}} \\\label{oprdn}
  & + \underbrace {\Pr \left( {\min \left( {\frac{{\gamma _{SR}^{{D_f}}}}{{{\gamma _{{{\rm{th}}_f}}}}},\frac{{\gamma _{SR}^{{D_n}}}}{{{\gamma _{{{\rm{th}}_n}}}}}} \right) \geqslant 1,\min \left( {\frac{{\gamma _{R{D_n}}^{{D_f}}}}{{{\gamma _{{{\rm{th}}_f}}}}},\frac{{\gamma _{R{D_n}}^{{D_n}}}}{{{\gamma _{{{\rm{th}}_n}}}}}} \right) < 1,\min \left( {\frac{{\gamma _{S{D_n}}^{{D_f}}}}{{{\gamma _{{{\rm{th}}_f}}}}},\frac{{\gamma _{S{D_n}}^{{D_n}}}}{{{\gamma _{{{\rm{th}}_n}}}}}} \right) < 1} \right)}_{{I_4}},
\end{align}
\!\!where ${\gamma _{{{\rm{th}}_n}}}$ denotes the SNR threshold of $x_n$, $I_3$ refers to Case \emph{iii)} and $I_4$ refers to Case \emph{iv)}. The exact outage probability expression of $D_n$ for the considered JCS NOMA system is given in the following theorem.
\begin{theorem}
The exact outage probability of the near IoT device is expressed as
\begin{align}\nonumber
  {\rm P}_{\rm{out}}^{{D_n}} &= \left( {1 - {e^{ - \frac{\theta }{{{\beta _{S{D_n}}}}}}}} \right)\left\{ {1 - \frac{{{\beta _{SR}}}}{{2{\beta _{RR}}\left( {{\beta _{SR}} + {\beta _{LI}}\omega {\gamma _r}\theta } \right)}}\sqrt {\frac{{\pi {\beta _{SR}}}}{{\delta {\gamma _r}\theta }}} } \right. \\\label{oprdn1}
   &\times \left. {{e^{\frac{{{\beta _{SR}}}}{{4\beta _{RR}^2\delta {\gamma _r}\theta }} - \frac{\theta }{{{\beta _{SR}}}} - \frac{\varphi }{{{\beta _{R{D_n}}}}}}}\left[ {1 - {\rm{erfc}}\left( {\frac{1}{{2{\beta _{RR}}}}\sqrt {\frac{{{\beta _{SR}}}}{{\delta {\gamma _r}\theta }}} } \right)} \right]} \right\},
\end{align}
where $\theta  \triangleq \max \left( {{\theta _1},{\theta _2}} \right)$, ${\theta _2} = {{{\gamma _{{{\rm{th}}_n}}}} \mathord{\left/
 {\vphantom {{{\gamma _{thn}}} {\left( {{a_n}{\gamma _c}} \right)}}} \right.
 \kern-\nulldelimiterspace} {\left( {{a_n}{\gamma _c}} \right)}}$, $\varphi  \triangleq \max \left( {{\varphi _1},{\varphi _2}} \right)$, ${\varphi _2} = {{{\gamma _{{{\rm{th}}_n}}}} \mathord{\left/
 {\vphantom {{{\gamma _{{{\rm{th}}_n}}}} {\left( {{a_n}{\gamma _r}} \right)}}} \right.
 \kern-\nulldelimiterspace} {\left( {{a_n}{\gamma _r}} \right)}}$. $\theta_1$ and $\varphi_1$ can be obtained from Theorem 1 with ${{a_f} > {a_n}{\gamma _{{{\rm{th}}_f}}}}$, otherwise ${\rm P}_{\rm{out}}^{{D_n}} = 1$.
\end{theorem}
\begin{proof}
Similar to Appendix A, substituting (\ref{sinrsdnf}), (\ref{sinrsdn}), and (\ref{sinryrf})-(\ref{sinrdnn}) into (\ref{oprdn}), using the PDFs and CDFs of $\rho _i$, (\ref{oprdn1}) can be obtained through a series of mathematical calculations.
\end{proof}
\begin{remark}
It can be observed from Theorem 1 and Theorem 2 that the outage probabilities are related to the paths between the nodes. For the direct links, the outage probabilities of the far and near IoT devices are respectively given by ${\rm P}_{\rm{out}}^{{D_f}} = 1 - \exp \left( { - {{{\theta _1}} \mathord{\left/
 {\vphantom {{{\theta _1}} {{\beta _{S{D_f}}}}}} \right.
 \kern-\nulldelimiterspace} {{\beta _{S{D_f}}}}}} \right)$ and ${\rm P}_{\rm{out}}^{{D_n}} = 1 - \exp \left( { - {\theta  \mathord{\left/
 {\vphantom {\theta  {{\beta _{S{D_n}}}}}} \right.
 \kern-\nulldelimiterspace} {{\beta _{S{D_n}}}}}} \right)$. For the cooperative links, the LSI and echo generate great influence on the outage probabilities of the IoT devices. It should be pointed that if the target is out of the sensing range, the system reduces to the traditional FD-NOMA scenario. In addition, when $\omega  = 0$, the MiBS reduces to the HD operation mode.
\end{remark}

\subsection{Diversity Order}
\vspace{-2 mm}
To obtain further insights for the considered cooperative NOMA scheme, the diversity orders of the IoT devices in the high SNR regime are analyzed in this subsection. The definition of diversity order is expressed as\cite{liumaeu}
\begin{equation}
\label{deford}
{D_O} =  - \mathop {\lim }\limits_{\Upsilon  \to \infty } \frac{{\log \left( {{\rm P}_{\rm{out}}^\infty } \right)}}{{\log \Upsilon }},
\end{equation}
where $\Upsilon  \in \left\{ {{\gamma _c},{\gamma _r}} \right\}$ denotes the transmitted SNR and ${\rm P}_{\rm{out}}^\infty $ is the asymptotic outage probability of the IoT devices in the high SNR regime.

The asymptotic outage probabilities of $D_f$ and $D_n$ in the high SNR regime for the considered JCS NOMA system are summarized in the following corollaries.
\vspace{-5 mm}
\begin{corollary}
The asymptotic outage probability of $D_f$ is expressed as
\begin{equation}\label{asdf2}
  {\rm{P}}_{\rm{out}}^{{D_f},\infty } \!\approx \!\frac{{{\theta _1}}}{{{\beta _{S{D_f}}}}}\!\left\{\! {1\! - \!\left( \!{1 \!- \!\frac{{{\varphi _1}}}{{{\beta _{R{D_f}}}}}} \right)\!\frac{{{\beta _{SR}}}}{{2{\beta _{RR}}\left( {{\beta _{SR}} \!+ \!{\beta _{LI}}\omega \theta _1^*} \right)}}\sqrt {\frac{{\pi {\beta _{SR}}}}{{\delta \theta _1^*}}} {e^{\frac{{{\beta _{SR}}}}{{4\beta _{RR}^2\delta \theta _1^*}}}}\!\left[ {1 \!-\! {\rm{erfc}}\!\left( {\frac{1}{{2{\beta _{RR}}}}\sqrt {\frac{{{\beta _{SR}}}}{{\delta \theta _1^*}}} } \!\right)} \!\right]} \!\right\},
\end{equation}
where $\theta _1^* = {{{\gamma _{{{\rm{th}}_f}}}} \mathord{\left/
 {\vphantom {{{\gamma _{thf}}} {\left( {{a_f} - {a_n}{\gamma _{thf}}} \right)}}} \right.
 \kern-\nulldelimiterspace} {\left( {{a_f} - {a_n}{\gamma _{{{\rm{th}}_f}}}} \right)}}$.
\end{corollary}
\vspace{-5 mm}
\begin{proof}
See Appendix B.
\end{proof}
By substituting (\ref{asdf2}) into (\ref{deford}), we can obtain that the diversity order of $D_f$ in the high SNR regime is $D_O^f = 1$.
\vspace{-5 mm}
\begin{corollary}
The asymptotic outage probability of $D_n$ is expressed as
\begin{align}\label{asdn2}
{\rm{P}}_{\rm{out}}^{{D_n},\infty } \!\approx\! \frac{\theta }{{{\beta _{S{D_n}}}}}\!\left\{ \!{1\! -\! \left( \!{1 \!-\! \frac{\varphi }{{{\beta _{R{D_n}}}}}} \! \right)\!\frac{{{\beta _{SR}}}}{{2{\beta _{RR}}\left( {{\beta _{SR}}\! +\! {\beta _{LI}}\omega {\theta ^*}} \right)}}\!\sqrt {\frac{{\pi {\beta _{SR}}}}{{\delta {\theta ^*}}}} {e^{\frac{{{\beta _{SR}}}}{{4\beta _{RR}^2\delta {\theta ^*}}}}}\!\left[ \!{1 \!-\! {\rm{erfc}}\!\left( \!{\frac{1}{{2{\beta _{RR}}}}\!\sqrt {\frac{{{\beta _{SR}}}}{{\delta {\theta ^*}}}} } \!\right)} \!\right]} \!\right\},
\end{align}
where ${\theta ^*} \triangleq \max \left( {\theta _1^*,\theta _2^*} \right)$ and $\theta _2^* = {{{\gamma _{{{\rm{th}}_n}}}} \mathord{\left/
 {\vphantom {{{\gamma _{{{\rm{th}}_n}}}} {{a_n}}}} \right.
 \kern-\nulldelimiterspace} {{a_n}}}$.
\end{corollary}
\vspace{-5 mm}
\begin{proof}
Similar to Appendix B, substituting (\ref{sinrsdnf}), (\ref{sinrsdn}) and (\ref{sinryrf})-(\ref{sinrdnn}) into (\ref{oprdn}), with the PDF and CDF of $\rho _i$, (\ref{asdn2}) can be obtained after some mathematical calculations.
\end{proof}
By substituting (\ref{asdn2}) into (\ref{deford}), the diversity order of $D_n$ in the high SNR regime is $D_O^n = 1$.
\vspace{-5 mm}
\begin{remark}
For the proposed cooperative NOMA scheme, the outage probabilities of cooperative links are constant in the high SNR regime. The diversity orders caused by cooperative links are zero due to the existence of LSI and echo. In this case, the diversity orders of the IoT devices are provided by the direct links. For the non-cooperative NOMA scheme, the asymptotic outage probabilities of the far and near IoT devices are respectively given by ${\rm P}_{\rm{out}}^{{D_f},\infty } \approx {{{\theta _1}} \mathord{\left/
 {\vphantom {{{\theta _1}} {{\beta _{S{D_f}}}}}} \right.
 \kern-\nulldelimiterspace} {{\beta _{S{D_f}}}}}$ and ${\rm P}_{\rm{out}}^{{D_n},\infty } \approx {\theta  \mathord{\left/
 {\vphantom {\theta  {{\beta _{S{D_n}}}}}} \right.
 \kern-\nulldelimiterspace} {{\beta _{S{D_n}}}}}$, while the corresponding diversity orders are $D_O^f = D_O^n = 1$.
\end{remark}
\vspace{-5 mm}
\section{Ergodic Rate Analysis}

In this section, the ergodic rates of the IoT devices are explored for the considered JCS NOMA system.

The achievable rates of $D_f$ and $D_n$ are respectively expressed as \cite{lix20}
\vspace{-3 mm}
\begin{align}\label{r2f}
{R_f} &= \underbrace {{{\log }_2}\left( {1 + \min \left( {\gamma _{SR}^{{D_f}},\gamma _{R{D_f}}^{{D_f}},\gamma _{R{D_n}}^{{D_f}}} \right)} \right)}_{R_f^{1}} + \underbrace {{{\log }_2}\left( {1 + \min \left( {\gamma _{S{D_f}}^{{D_f}},\gamma _{S{D_n}}^{{D_f}}} \right)} \right)}_{R_f^{2}},\\\label{r2n}
{R_n} &= \underbrace {{{\log }_2}\!\left( {1 \!+\! \min \left( {\gamma _{SR}^{{D_n}},\gamma _{R{D_n}}^{{D_n}}} \right)} \!\right)}_{R_n^{1}} \!+ \!\underbrace {{{\log }_2}\!\left( {1 \!+\! \gamma _{S{D_n}}^{{D_n}}} \right)}_{R_n^{2}}.
\end{align}
\vspace{-5 mm}

Then, the ergodic rates of $D_f$ and $D_n$ are respectively given as
\vspace{-1 mm}
\begin{align}\label{eraf1}
R_{\rm{ave}}^f &= R_{\rm{ave}}^{f,1} + R_{\rm{ave}}^{f,2} = \mathbb{E}\left[ {R_f^1} \right] + \mathbb{E}\left[ {R_f^2} \right],\\\label{eran1}
R_{\rm{ave}}^n& = R_{\rm{ave}}^{n,1} + R_{\rm{ave}}^{n,2} = \mathbb{E}\left[ {R_n^1} \right] + \mathbb{E}\left[ {R_n^2} \right].
\end{align}

\begin{theorem}
The ergodic rate of the far IoT device is expressed as
\vspace{-1 mm}
\begin{align}\nonumber
  R_{\rm{ave}}^f &=  \int_0^\infty  {\frac{{g_c^{\frac{3}{2}}\left( {{w_1}} \right){e^{\frac{{{\beta _{SR}}{g_c}\left( {{w_1}} \right)}}{{4\beta _{RR}^2\delta {\gamma _r}}} - \frac{1}{{{\beta _{SR}}{g_c}\left( {{w_1}} \right)}} - \frac{1}{{{\beta _{R{D_f}}}{g_r}\left( {{w_1}} \right)}} - \frac{1}{{{\beta _{R{D_n}}}{g_r}\left( {{w_1}} \right)}}}}\left[ {1 - {\rm{erfc}}\left( {\frac{1}{{2{\beta _{RR}}}}\sqrt {\frac{{{g_c}\left( {{w_1}} \right){\beta _{SR}}}}{{\delta {\gamma _r}}}} } \right)} \right]}}{{\left( {1 + {w_1}} \right)\left( {{\beta _{SR}}{g_c}\left( {{w_1}} \right) + {\beta _{LI}}\omega {\gamma _r}} \right)}}} d{w_1} \hfill \\\label{rate2f}
   &\times \frac{{{\beta _{SR}}}}{{2{\beta _{RR}}\ln 2}}\sqrt {\frac{{\pi {\beta _{SR}}}}{{\delta {\gamma _r}}}}+ \frac{1}{{\ln 2}}\int_0^\infty  {\frac{1}{{1 + {w_2}}}{e^{ - \frac{1}{{{\beta _{S{D_f}}}{g_c}\left( {{w_2}} \right)}} - \frac{1}{{{\beta _{S{D_n}}}{g_c}\left( {{w_2}} \right)}}}}d{w_2}},
\end{align}
where ${g_c}\left( w \right) = {{\left( {{a_f}{\gamma _c} - {a_n}{\gamma _c}w} \right)} \mathord{\left/
 {\vphantom {{\left( {{a_f}{\gamma _c} - {a_n}{\gamma _c}w} \right)} w}} \right.
 \kern-\nulldelimiterspace} w}$ and ${g_r}\left( w \right) = {{\left( {{a_f}{\gamma _r} - {a_n}{\gamma _r}w} \right)} \mathord{\left/
 {\vphantom {{\left( {{a_f}{\gamma _r} - {a_n}{\gamma _r}w} \right)} w}} \right.
 \kern-\nulldelimiterspace} w}$.

The ergodic rate of the near IoT device is expressed as
\vspace{-1 mm}
\begin{align}\nonumber
  R_{\rm{ave}}^n&=  \int_0^\infty  {\frac{{\frac{{q_c^{\frac{3}{2}}\left( {{w_3}} \right)}}{{\left( {{\beta _{SR}}{q_c}\left( {{w_3}} \right) + {\beta _{LI}}\omega {\gamma _r}} \right)}}{e^{\frac{{{\beta _{SR}}{q_c}\left( {{w_3}} \right)}}{{4\beta _{RR}^2\delta {\gamma _r}}} - \frac{1}{{{\beta _{SR}}{q_c}\left( {{w_3}} \right)}} - \frac{1}{{{\beta _{R{D_n}}}{q_r}\left( {{w_3}} \right)}}}}\left[ {1 - {\rm{erfc}}\left( {\frac{1}{{2{\beta _{RR}}}}\sqrt {\frac{{{q_c}\left( {{w_3}} \right){\beta _{SR}}}}{{\delta {\gamma _r}}}} } \right)} \right]}}{{1 + {w_3}}}} d{w_3} \hfill \\\label{rate2n}
  & \times \frac{{{\beta _{SR}}}}{{2{\beta _{RR}}\ln 2}}\sqrt {\frac{{\pi {\beta _{SR}}}}{{\delta {\gamma _r}}}}- \frac{1}{{\ln 2}}{e^{\frac{1}{{{\beta _{S{D_n}}}{a_n}{\gamma _c}}}}}{\rm{Ei}}\left( { - \frac{1}{{{\beta _{S{D_n}}}{a_n}{\gamma _c}}}} \right),
\end{align}
where ${q_c}\left( w \right) = {{{a_n}{\gamma _c}} \mathord{\left/
 {\vphantom {{{a_n}{\gamma _c}} w}} \right.
 \kern-\nulldelimiterspace} w}$ and ${q_r}\left( w \right) = {{{a_n}{\gamma _r}} \mathord{\left/
 {\vphantom {{{a_n}{\gamma _r}} w}} \right.
 \kern-\nulldelimiterspace} w}$.
\end{theorem}
\vspace{-3 mm}
\begin{proof}
See Appendix C.
\end{proof}
\vspace{-1 mm}
Thus, the ergodic sum rate of the IoT devices is expressed as
\begin{equation}
\label{sumr1}
R_{\rm{ave}}^{\rm{sum}} = R_{\rm{ave}}^f + R_{\rm{ave}}^n.
\end{equation}

From (\ref{rate2f}) and (\ref{rate2n}), we can observe that it is knotty to obtain the exact ergodic rates of the IoT devices, if not impossible. As a compromise, we analyze the system ergodic sum rate by adopting some approximation.

The approximate ergodic rates of $D_f$ and $D_n$ are given in the following corollary.
\begin{corollary}
The approximated ergodic rate of the far IoT device is expressed as
\begin{align}\label{aserf}
R_{\rm{ave}}^{ap,f}& \approx \frac{1}{{\ln 2}}\left[ {\ln \left( {1 + \min \left( {\frac{{{a_f}{\beta _{S{D_f}}}{\gamma _c}}}{{{a_n}{\beta _{S{D_f}}}{\gamma _c} + 1}},\frac{{{a_f}{\beta _{S{D_n}}}{\gamma _c}}}{{{a_n}{\beta _{S{D_n}}}{\gamma _c} + 1}}} \right)} \right)} \right. \hfill \\\nonumber
   &\left.+ \ln \left( {1 + \min \left( {\frac{{{a_f}{\beta _{SR}}{\gamma _c}}}{{{a_n}{\beta _{SR}}{\gamma _c} + 2\beta _{RR}^2\delta {\gamma _r} + {\beta _{LI}}\omega {\gamma _r} + 1}},\frac{{{a_f}{\beta _{R{D_f}}}{\gamma _r}}}{{{a_n}{\beta _{R{D_f}}}{\gamma _r} + 1}},\frac{{{a_f}{\beta _{R{D_n}}}{\gamma _r}}}{{{a_n}{\beta _{R{D_n}}}{\gamma _r} + 1}}} \right)} \right)\right].
\end{align}

The approximate ergodic rate of the near IoT device is expressed as
\begin{equation}\label{asern}
R_{\rm{ave}}^{ap,n} \!\approx\! \frac{1}{{\ln 2}}\!\left[ {\ln \!\left( {1\! + \!{a_n}{\beta _{S{D_n}}}{\gamma _c}} \right)\! +\! \ln\! \left(\! {1\! +\! \min \!\left( {\frac{{{a_n}{\beta _{SR}}{\gamma _c}}}{{2\beta _{RR}^2\delta {\gamma _r} \!+\! {\beta _{LI}}\omega {\gamma _r}\! +\! 1}},{a_n}{\beta _{R{D_n}}}{\gamma _r}}\! \right)} \!\right)}\! \right].
\end{equation}
\end{corollary}
Based on (\ref{aserf}) and (\ref{asern}), the ergodic sum rate of the IoT devices can be approximated as
\begin{equation}
\label{asersum}
R_{\rm{ave}}^{ap,\rm{sum}} \approx R_{\rm{ave}}^{ap,f} + R_{\rm{ave}}^{ap,n}.
\end{equation}

\begin{proof}
See Appendix D.
\end{proof}
\vspace{-2 mm}
\begin{remark}
From Theorem 3 and Corollary 3, we can observe that the sensing signal has a negative impact on the ergodic rates of the IoT devices. For the cooperative links, due to the echo and LSI, the ergodic rates of $D_n$ and $D_f$ respectively tend to be ${\log _2}\left[ {1 + {{{a_n}{\beta _{SR}}} \mathord{\left/
 {\vphantom {{{a_n}{\beta _{SR}}} {\left( {2\beta _{RR}^2\delta  + {\beta _{LI}}\omega } \right)}}} \right.
 \kern-\nulldelimiterspace} {\left( {2\beta _{RR}^2\delta  + {\beta _{LI}}\omega } \right)}}} \right]$ and ${\log _2}\left[ {1 + {{{a_f}{\beta _{SR}}} \mathord{\left/
 {\vphantom {{{a_f}{\beta _{SR}}} {\left( {{a_n}{\beta _{SR}} + 2\beta _{RR}^2\delta  + {\beta _{LI}}\omega } \right)}}} \right.
 \kern-\nulldelimiterspace} {\left( {{a_n}{\beta _{SR}} + 2\beta _{RR}^2\delta  + {\beta _{LI}}\omega } \right)}}} \right]$ in the high SNR regime. For the non-cooperative NOMA scheme, the sum rate reduces to ${R_{sum}} = R_f^2 + R_n^2$. Moreover, it is noteworthy that when $\omega=0$, the rates of the far and near IoT devices are respectively expressed as ${R_f} = {{R_f^1} \mathord{\left/
 {\vphantom {{R_f^1} 2}} \right.
 \kern-\nulldelimiterspace} 2} + R_f^2$ and ${R_n} = {{R_n^1} \mathord{\left/
 {\vphantom {{R_n^1} 2}} \right.
 \kern-\nulldelimiterspace} 2} + R_n^2$, where ${1 \mathord{\left/
 {\vphantom {1 2}} \right.
 \kern-\nulldelimiterspace} 2}$ indicates that the signals transmission needs two time slots for HD-NOMA system.
\end{remark}

\section{Sensing Probability Analysis}
\vspace{-1 mm}
For sensing, the fundamental task is to find the target and make the correct judgment. Generally, the received echo is a mixture of the desired signal, interference, and noise. The presence of the target can be determined by the power of the received signals. Assuming that $H_0$ denotes null hypothesis and $H_1$ includes echo, interference and AWGN, they can be defined as\footnote{During the sensing process, the optimal detector needs to know all the information of the target, which is overly idealistic in practice. Therefore, this paper designs a sub-optimal detector based on Likelihood Ratio Test (LRT) \cite{Maio20} to analyze the sensing performance of the MiBS.}
\begin{align}\nonumber
  &{H_0}:{y_{SR}} = {h_{SR}}{y_c} + {h_{LI}}\sqrt {\omega {P_{sen}}} {x_{LI}} + {n_{SR}} \hfill, \\\label{h01}
  &{H_1}:{y_{SR}} = {h_{SR}}{y_c} + {h_{RR}}\sqrt \delta  {y_r} + {h_{LI}}\sqrt {\omega {P_{sen}}} {x_{LI}} + {n_{SR}},
\end{align}
where $y_c$ and $y_r$ denote the signals transmitted by the MaBS and the MiBS.

If the power of interference and noise exceeds the preset threshold, it will lead to a misjudgment and the probability of such an event is classified as the false alarm probability. Correspondingly, the probability of sensing the target accurately is defined as sensing probability \cite{zhangxiand}. Obviously, the probability of target existence determined by the MiBS corresponds to the sensing probability in the case of $H_1$. On the contrary, in the case of $H_0$, the probability that the MiBS makes the same judgment is false alarm probability. From (\ref{h01}), it can be obtained that the received power of the MiBS is the non-central Chi-square random variable with four degree-of-freedoms (DoFs) under $H_0$ and five DoFs under $H_1$. Therefore, according to $H_0$, the false-alarm probability of the MiBS is expressed as
\vspace{-1 mm}
\begin{equation}
\label{fap}
{{\rm P}_{fa}} = {Q_{2}}\left( {\sqrt {\frac{{2{{\left\| {{h_{SR}}\left( {\sqrt {{a_n}{P_{com}}}  + \sqrt {{a_f}{P_{com}}} } \right) + {h_{LI}}\sqrt {\omega {P_{sen}}} } \right\|}^2}}}{{{N_0}}}} ,\sqrt {\frac{{2\zeta }}{{{N_0}}}} } \right),
\end{equation}
where $\zeta $ denotes the sensing threshold and $Q\left( { \cdot , \cdot } \right)$ represents the Marcum Q-function. The sensing probability of the MiBS according to $H_1$ can be expressed as
\vspace{-1 mm}
\begin{equation}
\label{senp}
{{\rm P}_d} = {Q_{\frac{5}{2}}}\left( {\sqrt {\frac{{2{{\left\| {{h_{SR}}\left( {\sqrt {{a_n}{P_{com}}}  + \sqrt {{a_f}{P_{com}}} } \right) + {h_{LI}}\sqrt {\omega {P_{sen}}}  + {h_{RR}}\sqrt {\delta {P_{sen}}} } \right\|}^2}}}{{{N_0}}}} ,\sqrt {\frac{{2\zeta }}{{{N_0}}}} } \right).
\end{equation}

In the following, we adopt the received signal of the MiBS to estimate the target information. According to (\ref{ysr}), the received power at the MiBS can be expressed as
\vspace{-3 mm}
\begin{equation}
\label{repower}
\mathbb{E}\left[ {{{\left| {{y_{SR}}} \right|}^2}} \right] = {\rho _{SR}}{P_{com}} + {\rho _{RR}}\delta {P_{sen}} + {\rho _{LI}}\omega {P_{sen}} + {N_0}.
\end{equation}
Obviously, $\mathbb{E}\left[ {{{\left| {{y_{SR}}} \right|}^2}} \right]$, $\rho_{SR}$, $P_{com}$, $P_{sen}$, and $N_0$ are known for the MiBS. Then the information related to the target can be expressed as
\vspace{-1 mm}
\begin{equation}
\label{rhorra}
{\rho _{RR}}\delta  = \frac{{\mathbb{E}\left[ {{{\left| {{y_{SR}}} \right|}^2}} \right] - \left( {{\rho _{LI}}\omega  + 1} \right){P_{sen}} - {\rho _{SR}}{P_{com}} - {N_0}}}{{{P_{sen}}}}.
\end{equation}

Note that during the signal sensing process, no matter whether the target exists in the sensing range, $\rho_{RR}\delta$ needs to be estimated by the received signal. Then the estimated $\rho_{RR}\delta$ should be re-substituted into the designed detector to determine whether there is a target. If the target is present and sensed by the MiBS, the probability corresponding to this occurrence is the sensing probability, shown as (\ref{senp}). If the target is absent while still sensed by the MiBS, the probability of such an event occurring is the false alarm probability, shown as (\ref{fap}). In this case, the estimated $\rho_{RR}\delta$ will cause a large deviation from the real value.

\begin{remark}
From (\ref{fap}), it can be observed that for the given false alarm probability, there exists a unique sensing threshold corresponding to it, and then the sensing probability can be obtained. In addition, we reiterate that if the target is outside the range of sensing, the obtained result of $\rho_{RR}\delta$ through (\ref{rhorra}) may be generated by the randomness of noise, thus further signal detection is essential to determine whether the target exists.
\end{remark}

\section{Performance Optimization}
\vspace{-1 mm}
Power allocation is essential to enhance the system performance and the resource utilization. Therefore, in this section, two OPA designs are proposed for the considered JCS NOMA system, namely, SCD, and CCD. Specifically, SCD aims to maximize the received SINR of the sensing signal at the MiBS with the IoT devices' QoS constraints and CCD is to maximize the IoT devices' sum rate with the constraints of minimum SINR requirements. It is worth stating that both SCD and CCD are of practical interests. According to (\ref{fap}) and (\ref{senp}), it can be obtained that increasing the SINR of the sensing signal can improve the received echo power, so as to enhance the sensing performance. This means that for SCD, increasing the received SINR of sensing signal at the MiBS can improve the sensing probability and guarantee the outage performance of the IoT devices for the cooperative links. Moreover, the proportion of the sensing signal is constant with the prescribed sensing probability, which indicates that for CCD, the sum rate is improved as well as meeting the sensing performance. Indeed, these OPA problems, from different perspectives, are meaningful to the JCS systems in practical scenarios. To this end, we first derive the received SINR of $x_r$ at the MiBS, which according to (\ref{ysr}) can be expressed as
\begin{equation}
\label{gammsr}
\gamma _{SR}^R = \frac{{\delta {\rho _{RR}}{\gamma _r}}}{{{\rho _{SR}}{\gamma _c} + {\rho _{LI}}\omega {\gamma _r} + 1}}.
\end{equation}

\vspace{-15 mm}
\subsection{SCD Scheme}
On the basis of (\ref{gammsr}), the SCD scheme can be formulated as
\vspace{-2 mm}
\begin{align}
\label{sinrm1}
  \mathop {\rm{maximize} }\limits_{{P_{com}},{P_{sen}},{a_n}}\;\;&\gamma _{SR}^R \hfill \\\tag{39a}
{\text{s}}{\text{.t}}{\text{.  }} &\min \left( {\gamma _{SR}^{{D_f}},\gamma _{R{D_f}}^{{D_f}},\gamma _{R{D_n}}^{{D_f}}} \right) \geqslant {\gamma _{{{\rm{th}}_f}}}, \hfill \\\tag{39b}
  &\min \left( {\gamma _{SR}^{{D_n}},\gamma _{R{D_n}}^{{D_n}}} \right) \geqslant {\gamma _{{{\rm{th}}_n}}} \hfill, \\\tag{39c}
 & 0 \leqslant {P_{com}} \leqslant {P_{\max }},\;{\text{0}} \!\leqslant {P_{sen}} \!\leqslant {P_{\max}}, \hfill \\\tag{39d}
  &{a_n} + {a_f} = 1,{\text{0}} < {a_n} < {a_f} < 1 ,
\end{align}
\!\!where (39a) indicates that $x_f$ can be decoded successfully over the links of $S \to R$, $R\to D_f$, and $R \to D_n$, (39b) indicates that $x_n$ can be decoded successfully over the links of $S \to R$ and $R\to D_n$, (39c) denotes the power budget of the MaBS and the MiBS and (39d) guarantees the enforceability of NOMA protocol.

Substituting (\ref{sinryrf})-(\ref{sinrff}) into (39a) and (39b), the optimization problem can be rewritten as
\begin{align}\label{osinrr}
  \mathop {\rm{maximize} }\limits_{{P_{com}},{P_{sen}},{a_n}}\; \;&\gamma _{SR}^R \hfill \\\tag{40a}
  {\text{s.t.}}&\;{P_{com}} \geqslant \left( {\delta {\rho _{RR}} + \omega {\rho _{LI}}} \right)\theta '{P_{sen}} + {N_0}\theta ', \hfill \\\tag{40b}
  &\;{P_{sen}} \geqslant {C_1}, \;\left( \text{39c} \right),\;\left( \text{39d} \right),
\end{align}
where $\theta'= {\theta ^*} \mathord{\left/
 {\vphantom {{{\theta ^*}} {{\rho _{SR}}}}} \right.
 \kern-\nulldelimiterspace} {\rho _{SR}}$, ${C_1} \triangleq \max \left( {{C_{11}},{C_{12}}} \right)$, ${C_{11}} = {{{N_0}{\theta ^*}} \mathord{\left/
 {\vphantom {{{N_0}{\theta ^*}} {{\rho _{R{D_n}}}}}} \right.
 \kern-\nulldelimiterspace} {{\rho _{R{D_n}}}}}$, and ${C_{12}} = {{{N_0}\theta _1^*} \mathord{\left/
 {\vphantom {{{N_0}\theta _1^*} {{\rho _{R{D_f}}}}}} \right.
 \kern-\nulldelimiterspace} {{\rho _{R{D_f}}}}}$, while $\theta _1^*$ and $\theta ^*$ can be obtained from Corollary 1 and Corollary 2, respectively.

The optimization problem and variables in (\ref{osinrr}) are coupled, and it is difficult to obtain their optimal values at the same time. Thus, we first optimize $P_{sen}$ with fixed $P_{com}$ and $a_n$, and then optimize $P_{com}$ with fixed $P_{sen}$ and $a_n$, finally optimize $a_n$ with fixed $P_{com}$ and $P_{sen}$ by the loop-iteration method.

First, the OPA coefficient $P_{sen}$ is derived with fixed $P_{com}$ and $a_n$. Substituting (\ref{gammsr}) into (\ref{osinrr}), we can observe that $\gamma _{SR}^R$ is a convex function as ${{{\partial ^2}\gamma _{SR}^R} \mathord{\left/
{\vphantom {{{\partial ^2}\gamma _{SR}^R} {\partial P_{sen}^2}}} \right.
\kern-\nulldelimiterspace} {\partial P_{sen}^2}} > 0$. Therefore, the Lagrangian function of (\ref{osinrr}) can be expressed as
\vspace{-1 mm}
\begin{align}\nonumber
 L_1\left( {{P_{sen}},{\lambda _{11}},{\lambda _{12}},{\lambda _{13}}} \right)& = \frac{{\delta {\rho _{RR}}{P_{sen}}}}{{{\rho _{SR}}{P_{com}} + {\rho _{LI}}\omega {P_{sen}} + 1}} + {\lambda _{11}}\left( {{P_{com}} - \left( {\delta {\rho _{RR}} + \omega {\rho _{LI}}} \right)\theta '{P_{sen}} - {N_0}\theta '} \right) \hfill \\\label{lg1}
  & + {\lambda _{12}}\left( {{P_{sen}} - {C_1}} \right) + {\lambda _{13}}\left( {{P_{\max }} - {P_{sen}}} \right),
\end{align}
\!\!where $\lambda _{11}$, $\lambda _{12}$, and $\lambda _{13}$ are the Lagrange multipliers with respect to constraints (40a), (40b), and (39c), respectively. The Karush-Kuhn-Tucker (KKT) condition applied for $L\left(  \cdot  \right)$ is expressed as
\begin{equation}\label{kkt11}
  \frac{{\partial {L_1}}}{{\partial {P_{sen}}}} = \frac{{\delta {\rho _{RR}}\left( {{\rho _{SR}}{P_{com}} + 1} \right)}}{{{{\left( {{\rho _{SR}}{P_{com}} + {\rho _{LI}}\omega {P_{sen}} + 1} \right)}^2}}} - {\lambda _{11}}\left( {\delta {\rho _{RR}} + \omega {\rho _{LI}}} \right)\theta ' + {\lambda _{12}} - {\lambda _{13}}=0.
\end{equation}

From (\ref{kkt11}), after some mathematical manipulations, the OPA coefficient $P_{sen}$ is expressed as
\begin{equation}\label{opsen1}
P_{sen}^* = {\left( {\frac{1}{{{\rho _{LI}}\omega }}\left[ {\sqrt {\frac{{\delta {\rho _{RR}}\left( {{\rho _{SR}}{P_{com}} + 1} \right)}}{{{\lambda _{11}}\left( {\delta {\rho _{RR}} + \omega {\rho _{LI}}} \right)\theta ' - {\lambda _{12}} + {\lambda _{13}}}}}  - \left( {{\rho _{SR}}{P_{com}} + 1} \right)} \right]} \right)^ + },
\end{equation}
where ${\left( \Delta  \right)^ + } \triangleq \max \left( {0,\Delta } \right)$.

Next, we optimize $P_{com}$ with fixed $P_{sen}$ and $a_n$. Obviously, the objective function in (\ref{osinrr}) is monotonically decreasing with respect to $P_{com}$. Herewith, according to (39c), (40a), and (\ref{opsen1}), the OPA parameter $P_{com}^*$ is expressed as
\begin{equation}
\label{opcom1}
P_{com}^* = \min \left( {\left( {\delta {\rho _{RR}} + \omega {\rho _{LI}}} \right)\theta 'P_{sen} + {N_0}\theta ',{P_{\max }}} \right).
\end{equation}

Finally, let us turn our attention to the optimization of $a_n$ with fixed $P_{com}$ and $P_{sen}$. From (\ref{osinrr}), we can observe that $\gamma_{SR}^R$ is independent of $a_n$, while $P_{com}^*$ is related to $a_n$. Furthermore, (\ref{opcom1}) indicates that $P_{com}^*$ needs be the minimum in the feasible regime. By exploiting this property, the optimization problem can be reformulated as
\begin{align}\label{oan}
  \mathop {\rm{minimize} }\limits_{{a_n}} &\;\left( {\delta {\rho _{RR}} + \omega {\rho _{LI}}} \right)\theta '{C_1} + {N_0}\theta ' \hfill \\\tag{45a}
  {\text{s.t.}}&\;\theta ' = \frac{{{\theta ^*}}}{{{\rho _{SR}}}},{\theta ^*} = \max \left( {\theta _1^*,\theta _2^*} \right),\left( \text{39d} \right)\!,
\end{align}
where $\theta _1^* = {{{\gamma _{{{\rm{th}}_f}}}} \mathord{\left/
 {\vphantom {{{\gamma _{{{\rm{th}}_f}}}} {\left( {{a_f} - {a_n}{\gamma _{{{\rm{th}}_f}}}} \right)}}} \right.
 \kern-\nulldelimiterspace} {\left( {{a_f} - {a_n}{\gamma _{{{\rm{th}}_f}}}} \right)}}$ and $\theta _2^* = {{{\gamma _{{{\rm{th}}_n}}}} \mathord{\left/
 {\vphantom {{{\gamma _{{{\rm{th}}_n}}}} {{a_n}}}} \right.
 \kern-\nulldelimiterspace} {{a_n}}}$. Using variable substitution, (\ref{oan}) can be reformulated as
\begin{align}\label{oanf}
  \mathop {\rm{minimize} }\limits_{{a_n}} \;&\max \left( {\theta _1^*,\theta _2^*} \right) \hfill \\\tag{46a}
  \text{s.t.}&\left( \text{{39d}} \right).
\end{align}

It can be observed that $\theta_1^*$ is monotonically increasing, while $\theta_2^*$ is monotonically decreasing with respect to $a_n$. Letting $\theta_1^*=\theta_2^*$, we can obtain $a_n^\dag  = {{{\gamma _{{{\rm{th}}_n}}}} \mathord{\left/
 {\vphantom {{{\gamma _{{{\rm{th}}_n}}}} {\left( {{\gamma _{{{\rm{th}}_f}}} + {\gamma _{{{\rm{th}}_n}}} + {\gamma _{{{\rm{th}}_f}}}{\gamma _{{{\rm{th}}_n}}}} \right)}}} \right.
 \kern-\nulldelimiterspace} {\left( {{\gamma _{{{\rm{th}}_f}}} + {\gamma _{{{\rm{th}}_n}}} + {\gamma _{{{\rm{th}}_f}}}{\gamma _{{{\rm{th}}_n}}}} \right)}}$, that yields
 \vspace{-3 mm}
\begin{equation}\label{thetax}
 \max \left( {\theta _1^*,\theta _2^*} \right) = \left\{ {\begin{array}{*{20}{c}}
  {\theta _2^*,{a_n} \in \left( {0,a_n^\dag } \right]} \\
  {\theta _1^*,{a_n} \in \left[ {a_n^\dag ,0.5} \right)}
\end{array}} \right..
\end{equation}

According to the increase-decrease characteristics of $\theta_1$ and $\theta_2$, $\max \left( {\theta _1^*,\theta _2^*} \right)$ can be minimized if and only if ${{a_n} = a_n^\dag }$. Therefore, the OPA coefficient $a_n^*$ is expressed as
\vspace{-3 mm}
\begin{equation}
\label{opan}
a_n^* = \min \left( {a_n^\dag ,0.5 - \sigma } \right) ,
\end{equation}
where $\sigma  >0$ and $\sigma  \to 0$.

\begin{remark}
The OPA factors can be solved iteratively by alternate optimization algorithm. The computational complexity of the algorithm is $O\left( {\left( {{K_1} + 2} \right){K_2}} \right)$, where $K_1$ is the number of iterations to find the optimal $\lambda _{11}^*$, $\lambda _{12}^*$, and $\lambda _{13}^*$, $K_2$ is the iteration number of the main loop. Furthermore, the proposed SCD scheme focuses on improving the sensing behavior while ensuring the communication performance, which can be applied to intersections, sidewalks and other scenarios with high pedestrian flow density.
\end{remark}

\subsection{CCD Scheme}

Based on (\ref{r2f}) and (\ref{r2n}), the CCD scheme can be expressed as the following mathematical expression
\begin{align}\label{mst1}
  \mathop {\text{maximize} }\limits_{{P_{com}},{P_{sen}},{a_f}}\;\;& {R_{sum}} \hfill \\\tag{49a}
  {\text{s}}{\text{.t}}{\text{.  }}&\min \left( {\gamma _{SR}^{{D_f}},\gamma _{R{D_f}}^{{D_f}},\gamma _{R{D_n}}^{{D_f}},\gamma _{S{D_f}}^{{D_f}},\gamma _{S{D_n}}^{{D_f}}} \right) \geqslant {\gamma _{{{\rm{th}}_f}}}, \hfill \\\tag{49b}
  &\min \left( {\gamma _{SR}^{{D_n}},\gamma _{R{D_n}}^{{D_n}},\gamma _{S{D_n}}^{{D_n}}} \right) \geqslant {\gamma _{{{\rm{th}}_n}}}, \hfill \\\tag{49c}
  &\gamma _{SR}^R \geqslant \kappa,\; \left(\rm{39c}\right), \;\left(\rm{39d}\right),
\end{align}
\!\!where (49a) represents $x_f$ can be decoded successfully over the links of $S \to R$, $R\to D_f$, $R \to D_n$, $S\to D_f$, and $S \to D_n$, (49b) indicates that $x_n$ can be decoded successfully over the links of $S \to R$, $R\to D_n$ and $S\to D_n$, (49c) is the constraint for the received SINR of $x_r$ at the MiBS, $\kappa$ denotes lower bound of the received SINR and $R_{sum}$ is expressed as\footnote{It should be noted that due to the intractability of the sum rate expression, it is extremely difficult to obtain the OPA factors, thus we focus on the approximation sum rate in the high SNR regime \cite{lixiot21} \cite{Abbasi20}.}
\begin{align}\nonumber
  {R_{sum}} &= {\log _2}\left( {1 + \min \left( {\frac{{{a_f}{\rho _{SR}}}}{{{a_n}{\rho _{SR}} + {\rho _{RR}}\delta  + {\rho _{LI}}\omega }},\frac{{{a_f}}}{{{a_n}}}} \right)} \right) + {\log _2}\left( {1 + \frac{{{a_f}}}{{{a_n}}}} \right) \\\label{rsum}
  & + {\log _2}\left( {1 + \min \left( {\frac{{{a_n}{\rho _{SR}}}}{{{\rho _{RR}}\delta  + {\rho _{LI}}\omega }},\frac{{{a_n}{\rho _{R{D_n}}}{P_{sen}}}}{{{N_0}}}} \right)} \right) + {\log _2}\left( {1 + \frac{{{a_n}{\rho _{S{D_n}}}{P_{com}}}}{{{N_0}}}} \right).
\end{align}

It can be observed ${{{a_f}{\rho _{SR}}} \mathord{\left/
{\vphantom {{{a_f}{\rho _{SR}}} {\left( {{a_n}{\rho _{SR}} + {\rho _{RR}}\delta  + {\rho _{LI}}\omega } \right)}}} \right.
\kern-\nulldelimiterspace} {\left( {{a_n}{\rho _{SR}} + {\rho _{RR}}\delta  + {\rho _{LI}}\omega } \right)}} \leqslant {{{a_f}} \mathord{\left/
{\vphantom {{{a_f}} {{a_n}}}} \right.
\kern-\nulldelimiterspace} {{a_n}}}$ is always satisfied from (\ref{rsum}), thus we can rewrite the sum rate as
\vspace{-3 mm}
\begin{equation}\label{sumr1}
{R_{sum}} \!=\! \frac{1}{{\ln 2}}\left[ {\ln \left( {{l_1}} \right) \!-\! \ln \left( {{f_1}} \right)\! -\! \ln \left( {{a_n}} \right)\! +\! \ln \left( {{f_3}} \right) - \ln \left( {{N_0}} \right)\! +\! \min \left( {\ln \left( {{f_1}} \right) \!-\! \ln \left( {{l_2}} \right),\ln \left( {{f_2}} \right) \!-\! \ln \left( {{N_0}} \right)} \right)} \right],
\end{equation}
where ${l_1} = {\rho _{SR}} + {\rho _{RR}}\delta  + {\rho _{LI}}\omega $, ${l_2} = {\rho _{RR}}\delta  + {\rho _{LI}}\omega $, ${f_1} = {a_n}{\rho _{SR}} + {\rho _{RR}}\delta  + {\rho _{LI}}\omega $, ${f_2} = {a_n}{\rho _{R{D_n}}}{P_{sen}} + {N_0}$ and ${f_3} = {a_n}{\rho _{S{D_n}}}{P_{com}} + {N_0}$.

Substituting (\ref{sinrsdnf})-(\ref{sinrff}) and (\ref{gammsr}) into (49a)-(49c), after some mathematical manipulations, (\ref{mst1}) can be rewritten as
\begin{align}\label{osr1}
  \mathop {\text{maximize} }\limits_{{P_{com}},{P_{sen}},{a_n}}\;\;& {R_{sum}} \hfill \\\tag{52a}
  {\text{s}}{\text{.t}}{\text{. }}&{l_4}{P_{sen}} - \kappa {\rho _{SR}}{P_{com}} - \kappa {N_0} \geqslant 0, \hfill \\\tag{52b}
  &{l_3}{P_{com}} - {l_2}{\gamma _{{\rm{th}}_f}}{P_{sen}} - {N_0}{\gamma _{{\rm{th}}_f}} \geqslant 0, \hfill \\\tag{52c}
  &{a_n}{\rho _{SR}}{P_{com}} - {l_2}{\gamma _{{\rm{th}}_n}}{P_{sen}} - {N_0}{\gamma _{{\rm{th}}_f}} \geqslant 0, \hfill \\\tag{52d}
  &{C_1} \leqslant {P_{com}} \leqslant {P_{\max }},{C_1} \leqslant {P_{sen}} \leqslant {P_{\max }},\;\left(\rm{39d}\right),
\end{align}
where ${l_3} = {a_f}{\rho _{SR}} - {a_n}{\rho _{SR}}{\gamma _{{\rm{th}}_f}}$, ${l_4} = \delta {\rho _{RR}} - \kappa {\rho _{LI}}\omega $, ${C_{11}} = {{{N_0}{\gamma _{{\rm{th}}_f}}} \mathord{\left/
 {\vphantom {{{N_0}{\gamma _{{\rm{th}}_f}}} {\left( {{a_f}{\rho _{S{D_f}}} - {a_n}{\rho _{S{D_f}}}{\gamma _{{\rm{th}}_f}}} \right)}}} \right.
 \kern-\nulldelimiterspace} {\left( {{a_f}{\rho _{S{D_f}}} - {a_n}{\rho _{S{D_f}}}{\gamma _{{\rm{th}}_f}}} \right)}}$, ${C_{12}} = {{{N_0}{\gamma _{{\rm{th}}_f}}} \mathord{\left/
 {\vphantom {{{N_0}{\gamma _{{\rm{th}}_f}}} {\left( {{a_f}{\rho _{S{D_n}}} - {a_n}{\rho _{S{D_n}}}{\gamma _{{\rm{th}}_n}}} \right)}}} \right.
 \kern-\nulldelimiterspace} {\left( {{a_f}{\rho _{S{D_n}}} - {a_n}{\rho _{S{D_n}}}{\gamma _{{\rm{th}}_n}}} \right)}}$, ${C_{13}} = {{{N_0}{\gamma _{{\rm{th}}_n}}} \mathord{\left/
 {\vphantom {{{N_0}{\gamma _{{\rm{th}}_n}}} {\left( {{a_n}{\rho _{S{D_n}}}} \right)}}} \right.
 \kern-\nulldelimiterspace} {\left( {{a_n}{\rho _{S{D_n}}}} \right)}}$, and ${C_1} = \max \left( {{C_{11}},{C_{12}},{C_{13}}} \right)$.

Similar to (\ref{osinrr}), it is hard to optimize $P_{com}$, $P_{sen}$, and $a_n$ simultaneously, therefore we first optimize $P_{com}$ and $P_{sen}$ with fixed $a_n$ and then optimize $a_n$ with fixed $P_{com}$ and $P_{sen}$ by the loop-iteration method.

By introducing slack variable $V$ and according to (\ref{sumr1}), we can rewrite (\ref{osr1}) as
\vspace{-3 mm}
\begin{align}\label{opv}
  \mathop {\text{maximize} }\limits_{{P_{com}},{P_{sen}},V}\;\; &V \hfill \\\tag{53a}
{\text{s}}{\text{.t}}{\text{.  }} & V \leqslant  - \ln \left( {{a_n}} \right) + \ln \left( {{f_3}} \right) - \ln \left( {{l_2}} \right), \hfill \\\tag{53b}
 & V \leqslant \! - \ln \left( {{f_1}} \right) \! -\!  \ln \left( {{a_n}} \right) \! +\!  \ln \left( {{f_3}} \right) \! +\!  \ln \left( {{f_2}} \right) \! -\!  \ln \left( {{N_0}} \right), \hfill \\\tag{53c}
  &{l_4}{P_{sen}} - \kappa {\rho _{SR}}{P_{com}} - \kappa {N_0} \geqslant 0, \;\left(\text{52d}\right).
\end{align}
\vspace{-1 mm}
Since (\ref{opv}) is a convex optimization problem, the Lagrangian function can be expressed as
\vspace{-3 mm}
\begin{align}\nonumber
 & L\left( {{P_{com}},{P_{sen}},V,{\lambda _{21}},{\lambda _{22}},{\lambda _{23}},{\lambda _{24}},{\lambda _{25}},{\lambda _{26}},{\lambda _{27}},{\lambda _{28}},{\lambda _{29}}} \right) = V + {\lambda _{21}}\left( { - \ln \left( {{a_n}} \right) + \ln \left( {{f_3}} \right) - \ln \left( {{l_2}} \right) - V} \right) \hfill \\\nonumber
  & + {\lambda _{22}}\left( { - \ln \left( {{f_1}} \right) + \ln \left( {{f_3}} \right)} \right.\left. { - \ln \left( {{a_n}} \right) + \ln \left( {{f_2}} \right) - \ln \left( {{N_0}} \right) - V} \right) + {\lambda _{23}}\left( {{l_4}{P_{sen}} - \kappa {\rho _{SR}}{P_{com}} - \kappa {N_0}} \right) \hfill \\\nonumber
  & + {\lambda _{24}}\left( {{l_3}{P_{com}} - {l_2}{\gamma _{{\rm{th}}_f}}{P_{sen}} - {N_0}{\gamma _{{\rm{th}}_f}}} \right) + {\lambda _{25}}\left( {{a_n}{\rho _{SR}}{P_{com}} - {l_2}{\gamma _{{\rm{th}}_n}}{P_{sen}} - {N_0}{\gamma _{{\rm{th}}_f}}} \right) \hfill \\\label{l2o}
  & + {\lambda _{26}}\left( {{P_{\max }} - {P_{com}}} \right) + {\lambda _{27}}\left( {{P_{\max }} - {P_{sen}}} \right) + {\lambda _{28}}\left( {{P_{com}} - {C_1}} \right) + {\lambda _{29}}\left( {{P_{sen}} - {C_1}} \right).
\end{align}
It is easy to obtain that ${\lambda _{21}} + {\lambda _{22}} = 1$ and the KKT conditions are expressed as
\begin{align}\label{l2o1}
&\frac{{\partial L}}{{\partial {P_{com}}}}= \frac{{{a_n}{\rho _{S{D_n}}}}}{{{f_3}}} - {\lambda _{23}}\kappa {\rho _{SR}} + {\lambda _{24}}{l_3} + {\lambda _{25}}{a_n}{\rho _{SR}} - {\lambda _{26}} +  {\lambda _{28}} = 0,\\\label{l2o2}
 & \frac{{\partial L}}{{\partial {P_{sen}}}} = \frac{{{\lambda _{22}}{a_n}{\rho _{R{D_n}}}}}{{{f_2}}} + {\lambda _{23}}{l_4} - {\lambda _{24}}{l_2}{\gamma _{thf}} - {\lambda _{25}}{l_2}{\gamma _{thn}} - {\lambda _{27}} + {\lambda _{29}} = 0 .
\end{align}
Based on (\ref{l2o1}) and (\ref{l2o2}), the OPA coefficients $P_{com}$ and $P_{sen}$ are expressed as
\begin{align}\label{opcom2}
P_{com}^* &= \left[\frac{1}{{{\lambda _{23}}\kappa {\rho _{SR}} - {\lambda _{24}}{l_3} - {\lambda _{25}}{a_n}{\rho _{SR}} + {\lambda _{26}} - {\lambda _{28}}}} - \frac{{{N_0}}}{{{a_n}{\rho _{S{D_n}}}}}\right]^+,\\\label{opsen2}
P_{sen}^* &= \left[\frac{{{\lambda _{22}}}}{{{\lambda _{24}}{l_2}{\gamma _{{\rm{th}}_f}} - {\lambda _{23}}{l_4} + {\lambda _{25}}{l_2}{\gamma _{{\rm{th}}_n}} + {\lambda _{27}} - {\lambda _{29}}}} - \frac{{{N_0}}}{{{a_n}{\rho _{S{D_n}}}}}\right]^+.
\end{align}

Next, we optimize $a_n$ with fixed $P_{com}$ and $P_{sen}$. Substituting (\ref{sinrsdnf})-(\ref{sinrff}) into (52a) and (52b), after some mathematical manipulations, the optimization problem can be written as
\begin{align}\label{oan2}
  \mathop {\text{maximize} }\limits_{{a_n}} &\;R_{sum} \hfill \\\tag{59a}
{\text{s}}{\text{.t}}{\text{.  }}&{C_2} \leqslant {a_n} \leqslant {C_3},\;\left(\text{45e}\right),
\end{align}
where ${C_2} = \max \left( {{C_{21}},{C_{22}},{C_{23}}} \right)$, ${C_{21}} = {{\left( {{l_2}{P_{sen}}{\gamma _{{\rm{th}}_n}} + {N_0}{\gamma _{{\rm{th}}_n}}} \right)} \mathord{\left/
 {\vphantom {{\left( {{l_2}{P_{sen}}{\gamma _{{\rm{th}}_n}} + {N_0}{\gamma _{{\rm{th}}_n}}} \right)} {\left( {{\rho _{SR}}{P_{com}}} \right)}}} \right.
 \kern-\nulldelimiterspace} {\left( {{\rho _{SR}}{P_{com}}} \right)}}$, ${C_{22}} = {{{N_0}{\gamma _{{\rm{th}}_n}}} \mathord{\left/
 {\vphantom {{{N_0}{\gamma _{{\rm{th}}_n}}} {\left( {{\rho _{S{D_n}}}{P_{com}}} \right)}}} \right.
 \kern-\nulldelimiterspace} {\left( {{\rho _{S{D_n}}}{P_{com}}} \right)}}$, ${C_{23}} = {{{N_0}{\gamma _{{\rm{th}}_n}}} \mathord{\left/
 {\vphantom {{{N_0}{\gamma _{{\rm{th}}_n}}} {\left( {{\rho _{R{D_n}}}{P_{sen}}} \right)}}} \right.
 \kern-\nulldelimiterspace} {\left( {{\rho _{R{D_n}}}{P_{sen}}} \right)}}$, ${C_3} = \min \left( {{C_{31}},{C_{32}},{C_{33}},{C_{34}},{C_{35}}} \right)$, ${C_{31}} = {g_1}\left( {{\rho _{S{D_f}}}} \right)$, ${C_{32}} = {g_1}\left( {{\rho _{S{D_n}}}} \right)$, ${C_{33}} = {g_2}\left( {{\rho _{R{D_f}}}} \right)$, ${C_{34}} = {g_2}\left( {{\rho _{R{D_n}}}} \right)$, ${C_{35}} = {{\left( {{P_{com}}{\rho _{SR}} - {l_2}{P_{sen}}{\gamma _{{\rm{th}}_f}} - {N_0}{\gamma _{{\rm{th}}_f}}} \right)} \mathord{\left/
 {\vphantom {{\left( {{P_{com}}{\rho _{SR}} - {l_2}{P_{sen}}{\gamma _{{\rm{th}}_f}} - {N_0}{\gamma _{{\rm{th}}_f}}} \right)} {\left[ {{P_{com}}{\rho _{SR}}\left( {1 + {\gamma _{{\rm{th}}_f}}} \right)} \right]}}} \right.
 \kern-\nulldelimiterspace} {\left[ {{P_{com}}{\rho _{SR}}\left( {1 + {\gamma _{{\rm{th}}_f}}} \right)} \right]}}$. ${g_1}\left( \rho  \right)$ and ${g_2}\left( \rho  \right)$ are respectively given by
\begin{equation}\label{g1rho}
{g_1}\left( \rho  \right) = \frac{{{P_{com}}\rho  - {N_0}{\gamma _{{\rm{th}}_f}}}}{{{P_{com}}\rho \left( {1 + {\gamma _{{\rm{th}}_f}}} \right)}},
\end{equation}
\begin{equation}\label{g1rho}
{g_2}\left( \rho  \right) = \frac{{{P_{sen}}\rho  - {N_0}{\gamma _{{\rm{th}}_f}}}}{{{P_{sen}}\rho \left( {1 + {\gamma _{{\rm{th}}_f}}} \right)}}.
\end{equation}

For the sake of analysis, $R_{sum}$ is divided into two cases as
\vspace{-1 mm}
\begin{equation}\label{rsan}
{R_{sum}} \!= \!\left\{ {\begin{array}{*{20}{c}}
  {\frac{1}{{\ln 2}}\left[ {\ln \left( {{l_1}} \right)\! -\! \ln \left( {{a_n}} \right)\! +\! \ln \left( {{f_3}} \right) \!- \!\ln \left( {{l_2}} \right) \!-\! \ln \left( {{N_0}} \right)} \right],\;\text{if}\;\ln \left( {{f_1}} \right) \!-\! \ln \left( {{l_2}} \right)\! \leqslant \!\ln \left( {{f_2}} \right) \!-\! \ln \left( {{N_0}} \right)} \\
  {\frac{1}{{\ln 2}}\left[ {\ln \left( {{l_1}} \right) \!-\! \ln \left( {{a_n}} \right)\! + \!\ln \left( {{f_3}} \right) \!+ \!\ln \left( {{f_2}} \right) \!- \!2\ln \left( {{N_0}} \right)} \right],\;\text{if}\;\ln \left( {{f_1}} \right) \!- \!\ln \left( {{l_2}} \right)\! > \!\ln \left( {{f_2}} \right) \!-\! \ln \left( {{N_0}} \right)}
\end{array}} \right.\!.
\end{equation}
If $\ln \left( {{f_1}} \right) - \ln \left( {{l_2}} \right) \leqslant \ln \left( {{f_2}} \right) - \ln \left( {{N_0}} \right)$, the first-order derivative of $R_{sum}$ with respect of $a_n$ is expressed as
\begin{equation}\label{der1}
\frac{{\partial {R_{sum}}}}{{\partial {a_n}}} = \frac{{{\rho _{S{D_n}}}{P_{com}}}}{{{a_n}{\rho _{S{D_n}}}{P_{com}} + {N_0}}} - \frac{1}{{{a_n}}}.
\end{equation}
Obviously, ${{\partial {R_{sum}}} \mathord{\left/
 {\vphantom {{\partial {R_{sum}}} {\partial {a_n}}}} \right.
 \kern-\nulldelimiterspace} {\partial {a_n}}} \leqslant 0$ is satisfied.

If $\ln \left( {{f_1}} \right) - \ln \left( {{l_2}} \right) > \ln \left( {{f_2}} \right) - \ln \left( {{N_0}} \right)$, then ${\rho _{SR}}{N_0} > {\rho _{R{D_n}}}{P_{sen}}\left( {{\rho _{RR}}\delta  + {\rho _{LI}}\omega } \right)$, and the first-order derivative of $R_{sum}$ with respect of $a_n$ is expressed as
\begin{equation}\label{der2}
  \frac{{\partial {R_{sum}}}}{{\partial {a_n}}} =  - \frac{{{\rho _{SR}}}}{{{a_n}{\rho _{SR}} + {\rho _{RR}}\delta  + {\rho _{LI}}\omega }} - \frac{1}{{{a_n}}} \hfill + \frac{{{\rho _{S{D_n}}}{P_{com}}}}{{{a_n}{\rho _{S{D_n}}}{P_{com}} + {N_0}}} + \frac{{{\rho _{R{D_n}}}{P_{sen}}}}{{{a_n}{\rho _{R{D_n}}}{P_{sen}} + {N_0}}}.
\end{equation}
Thus, ${{\partial {R_{sum}}} \mathord{\left/
 {\vphantom {{\partial {R_{sum}}} {\partial {a_n}}}} \right.
 \kern-\nulldelimiterspace} {\partial {a_n}}} \leqslant 0$ is tenable for the following reasons
\begin{align}\label{r1}
\frac{{{\rho _{S{D_n}}}{P_{com}}}}{{{a_n}{\rho _{S{D_n}}}{P_{com}} + {N_0}}} - \frac{1}{{{a_n}}} &\leqslant 0,\\\label{r2}
\frac{{{\rho _{R{D_n}}}{P_{sen}}}}{{{a_n}{\rho _{R{D_n}}}{P_{sen}} + {N_0}}} - \frac{{{\rho _{SR}}}}{{{a_n}{\rho _{SR}} + {\rho _{RR}}\delta  + {\rho _{LI}}\omega }} &\leqslant 0.
\end{align}

Therefore, it can be concluded that $R_{sum}$ is a monotonously decreasing function of $a_n$, and according to (55a) the OPA coefficient $a_n^*$ is given by
\begin{equation}\label{an2}
a_n^* = \min \left( {{C_2},0.5 - \sigma } \right),
\end{equation}
where $\sigma > 0$ and $\sigma \to 0$.

\begin{remark}
Similar to the SCD scheme, by using the alternate optimization algorithm, the OPA factors can be obtained. The computational complexity of the algorithm is $O\left( {\left( {{K_3} + 1} \right){K_4}} \right)$, where $K_3$ denotes the number of iterations to find the optimal Lagrangian, $K_4$ is the iteration number of the main loop. In addition, the CCD scheme is designed to improve the communication performance while ensuring sensing requirement, which is suitable for sparsely populated areas.
\end{remark}

\section{Numerical Results}

In this section, numerical results are provided to verify the theoretical analysis presented in Section III-Section VI by using Monte-Carlo computer simulation. Hereinafter, unless otherwise mentioned, the simulation parameters are considered as follows: $a_f=0.7$, $a_n=0.3$, ${\gamma _{{\rm{th}}_f}}=1$, ${\gamma _{{\rm{th}}_n}}=2$, $\delta=0.2$, $N_0=1$, $\rho_{LI}=-25$ dB, $\Omega =5$, $\alpha=4$, $d_{SR}= 10$ m, $d_{RD_f}=20$ m, $d_{RD_n}=15$ m, $d_{SD_n}=20$ m, $d_{SD_f}=25$ m, $d_{RT}=d_{TR}=12$ m.

\begin{figure} [!t]
\centering
\vspace{-1.2 cm}  
\includegraphics[width=8cm]{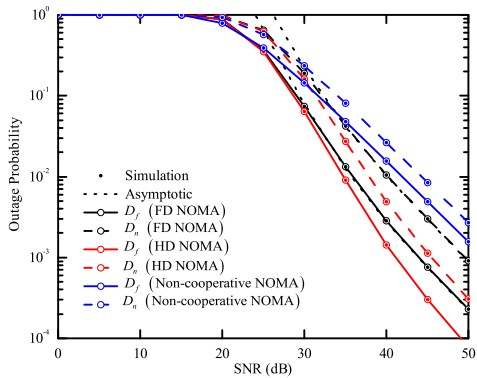}
\vspace{-0.8 cm}  
 \caption{Outage probabilities of  IoT devices versus the transmit SNR of $S$.}
\end{figure}
Fig. 2 presents the outage probability of $D_n$ and $D_f$ versus the transmit SNR of $S$. For comparison, the outage performance of the traditional non-cooperative NOMA and cooperative HD-NOMA schemes are presented as well.\footnote{It should be described that the total power is consistent across all the considered schemes to ensure a fair comparison.} It can be seen from Fig. 2 that the Monte Carlo simulation matches the analytical results perfectly and the convergence between the curves of asymptotic and analytical in the high SNR regime verify the accuracy of the derived results in (\ref{oprf1}), (\ref{oprdn1}), (\ref{asdf2}), and (\ref{asdn2}). From Fig. 2, we can also observe that on the premise of realizing the sensing function, the outage probabilities of the IoT devices for the proposed JCS NOMA scheme outperform the traditional non-cooperative NOMA with the same energy consumption, which is due to the fact that cooperative NOMA transmission can improve the frequency band utilization significantly. Furthermore, we notice that the outage performance of the proposed FD-NOMA scheme is worse compared with that of HD-NOMA, mainly due to the negative impact of LSI on the cooperative communication links. In addition, it can be obtained that the outage performance of $D_f$ is better than $D_n$, which is due to more power is allocated to $D_f$.

\begin{figure} [!t]
\centering
\includegraphics[width=8cm]{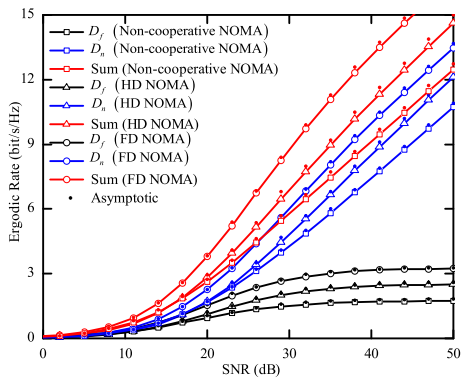}
\vspace{-0.8 cm}
 \caption{Ergodic rates of IoT devices versus the transmit SNR of $S$.}
 \vspace{-1 cm}  
\end{figure}
Fig. 3 demonstrates the ergodic rates of $D_n$ and $D_f$ versus the transmit SNR of $S$ in FD and HD mode. We reiterate that the total power consumption of three schemes is identical to ensure a fair comparison. It can be observed from Fig. 3 that the ergodic rate of $D_n$ keeps increasing, which is due to that the increases of the communication transmit power can improve the received SINR of the near IoT device. While for $D_f$, the ergodic rate tends to be standard eventually. Exactly, the ergodic rates of the far IoT device are numerically expressed as ${\log _2}\left( {1 + {{{a_f}} \mathord{\left/
 {\vphantom {{{a_f}} {{a_n}}}} \right.
 \kern-\nulldelimiterspace} {{a_n}}}} \right) + {\log _2}\left( {1 + {{{a_f}{\beta _{SR}}} \mathord{\left/
 {\vphantom {{{a_f}{\beta _{SR}}} {\left( {{a_n}{\beta _{SR}} + 2\beta _{RR}^2\delta  + {\beta _{LI}}\omega } \right)}}} \right.
 \kern-\nulldelimiterspace} {\left( {{a_n}{\beta _{SR}} + 2\beta _{RR}^2\delta  + {\beta _{LI}}\omega } \right)}}} \right)$ (FD-NOMA), ${\log _2}\left( {1 + {{{a_f}} \mathord{\left/
 {\vphantom {{{a_f}} {{a_n}}}} \right.
 \kern-\nulldelimiterspace} {{a_n}}}} \right) + 0.5{\log _2}\left( {1 + {{{a_f}{\beta _{SR}}} \mathord{\left/
 {\vphantom {{{a_f}{\beta _{SR}}} {\left( {{a_n}{\beta _{SR}} + 2\beta _{RR}^2\delta  + {\beta _{LI}}\omega } \right)}}} \right.
 \kern-\nulldelimiterspace} {\left( {{a_n}{\beta _{SR}} + 2\beta _{RR}^2\delta  + {\beta _{LI}}\omega } \right)}}} \right)$ (HD-NOMA), ${\log _2}\left( {1 + {{{a_f}} \mathord{\left/
 {\vphantom {{{a_f}} {{a_n}}}} \right.
 \kern-\nulldelimiterspace} {{a_n}}}} \right)$ (non-cooperative NOMA) in the high SNR regime. Moreover, it can be seen that the ergodic rate of the near IoT device plays a decisive role in the ergodic sum rate in the high SNR regime.

 \begin{figure} [!t]
\centering
\vspace{-0.8 cm}
\includegraphics[width=8cm]{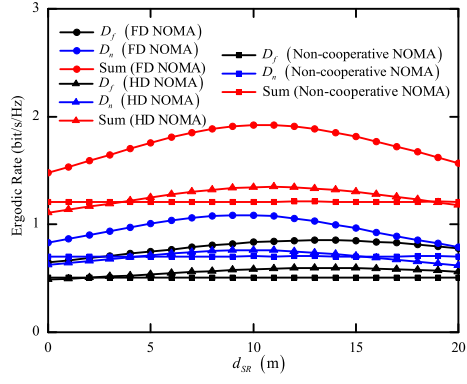}
\vspace{-0.7 cm}
\caption{Ergodic rates of IoT devices versus the relay location.}
\vspace{-1.0 cm}
\end{figure}
Fig. 4 illustrates ergodic rates of $D_n$ and $D_f$ versus the relay location in FD, HD, and non-cooperative schemes. Since the near IoT device makes a significant contribution to the ergodic sum rate, to better show the system performance, we assume that $S$, $R$, and $D_n$ are on a straight line, and $R$ is located between $S$ and $D_n$. We set $P_{com}=P_{sen}=20$ dB, $d_{SD_n}=5$ m, $d_{SD_f}=6$ m, the angle between $S\to D_f$ and $S\to D_n$ is set as $\phi=30^\circ$, thus the distance between $R$ and $D_f$ is ${d_{R{D_f}}} = \sqrt {d_{SR}^2 + d_{S{D_f}}^2 - 2{d_{SR}}{d_{S{D_f}}}\cos \phi } $. From Fig. 4, we can see that for the cooperative NOMA schemes, the ergodic sum rate of the IoT devices improves first and then decreases with the increase of $d_{SR}$. The reason for this phenomenon can be attributed to that the increase of $d_{SR}$ can enhance the channel gain between the relay and the IoT devices. However, with a shorter $d_{RD_n}$, the path loss between the MaBS and the relay increases, thus weakening the data transmission. Additionally, it can also be obtained that if the position of $R$ and $D_n$ is close to each other particularly, the ergodic rates of $D_n$ for the FD and HD NOMA are inferior to that of non-cooperative NOMA, since cooperative NOMA schemes need to share a portion of the total power for target sensing.

\begin{figure} [!t]
\centering
\includegraphics[width=8cm]{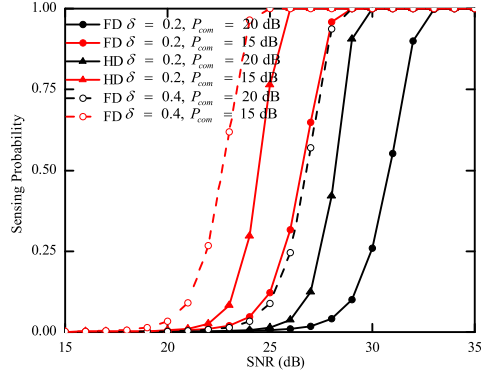}
\vspace{-0.8 cm}
\caption{Sensing probability versus the transmit SNR of the MiBS.}
\vspace{-0.5 cm}
\end{figure}
Fig. 5 shows the sensing probability of the MiBS versus the transmit SNR of $R$ in FD and HD modes with ${{\rm P}_{fa}} = {10^{ - 5}}$. Intuitively, the sensing probability can be improved by increasing the transmit SNR of $R$. We can further explore from Fig. 5 that for a fixed $P_{com}$, if $P_{sen} \leqslant P_{com}$, the sensing probability is approximately 0 and reducing $P_{com}$ can improve the sensing behavior, which is due to that the signal transmitted by the MaBS is regarded as interference in the considered system. Furthermore, due to the existence of LSI, the sensing performance of the FD NOMA scheme lags behind HD-NOMA with the same parameter settings, which is consistent with the result of (\ref{senp}). Therefore, in practical JCS system design, the operation mode of the relay should be considered reasonably according to the different requirements of sensing and communication.

\begin{figure} [!t]
\centering
\vspace{-0.5 cm}
\includegraphics[width=8cm]{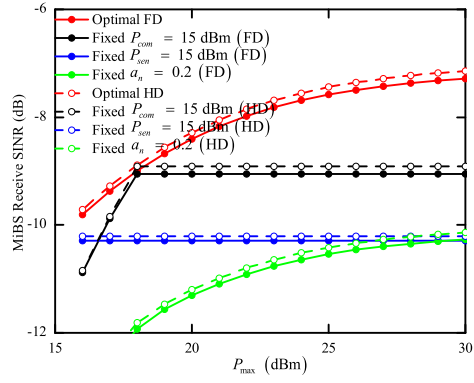}
\vspace{-0.8 cm}
\caption{Received SINR of $x_r$ versus the power budget with different power allocation schemes.}
\end{figure}
Fig. 6 plots the received SINR of $x_r$ at the MiBS versus the power budget $P_{\max}$ in both FD and HD modes. For comparison, we plot the curves of the proposed OPA scheme and other three random power allocation (RPA) baseline schemes, i.e., \emph{1) Jointly optimize $P_{sen}$ and $a_n$ with fixed $P_{com}$ = 15 dBm; \emph{2)} Jointly optimize $P_{com}$ and $a_n$ with fixed $P_{sen}$ = 15 dBm; \emph{3)} Jointly optimize $P_{com}$ and $P_{sen}$ with fixed $a_n$ = 0.2. For the fixed $P_{com}$, the received SINR of $x_r$ first increases with $P_{\max}$ and then with a diminishing return for an exceedingly large $P_{\max}$. This is because the received SINR of $x_r$ is an increasing function with respect to $P_{sen}$, and $P_{sen}$ can be increased by expanding the power budget $P_{\max}$, while the decoding performance of communication signals will be reduced on the contrary. When the minimal requirements of the communication signal decoding is reached, the received SINR of $x_r$ remains unchanged. For the fixed $P_{sen}$, the received SINR of $x_r$ is a monotonically decreasing function with respect to $P_{com}$, while increasing $P_{\max}$ cannot reduce $P_{com}$, thus the received SINR remains constant. For the fixed $a_n$, the received SINR of $x_r$ increases with $P_{\max}$, while due to the restraint of signal decoding requirements, the received SINR  lags behind the OPA scheme. In addition, it can be observed from Fig. 6 that due to the existence of LSI, the received SINR in FD mode is lower than that in HD mode.}

\begin{figure} [!t]
\centering
\includegraphics[width=8cm]{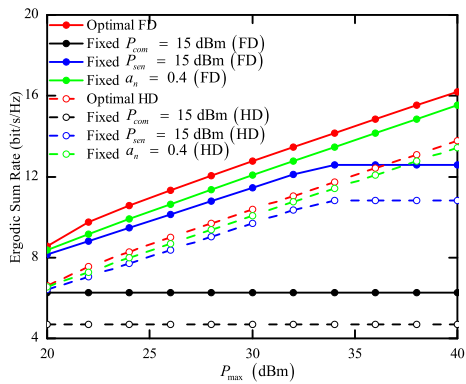}
\vspace{-0.8 cm}
\caption{Ergodic sum rate versus the power budget with different power allocation schemes.}
\vspace{-0.6 cm}
\end{figure}
Fig. 7 shows the advantages of the proposed OPA scheme over other RPA baseline schemes in terms of the ergodic sum rate in both FD and HD modes. As can be obtained from Fig. 7, the proposed OPA scheme is superior to the schemes of the fixed $P_{com}$, $P_{sen}$ or $a_n$. For the fixed $P_{com}$, due to the sensing constraint and the requirements of communication signals decoding, the ergodic sum rate remains constant. Besides, for a given $P_{sen}$, the ergodic sum rate increases with the increase of $P_{\max}$. This phenomenon continues until the communication reaches the peak power with the restraint of sensing demand. For the fixed $a_n$, restricted by the characteristics of the objective function and the decoding of communication signals, the ergodic sum rate is inferior to the proposed OPA scheme.
\vspace{-3 mm}
\section{Conclusion}
In this paper, we studied a JCS FD-NOMA system, in which the MiBS can realize two functions simultaneously, i.e., target sensing and cooperative relaying. We characterized the performance of communication and sensing. The exact and asymptotic outage probabilities, diversity orders, and approximate ergodic rates of the IoT devices were analyzed for communication, and the sensing probability of the MaBS was derived for sensing. To further consolidate the performance of the considered system, SCD, and CCD, two OPA problems were proposed. Then, the maximization problems of the sensing signal received SINR at the MiBS and IoT devices' ergodic sum rate were solved by using the Lagrange method and function monotonicity. The simulation results demonstrated that the outage performance of FD-NOMA system outperformed traditional non-cooperative in the context of constant total power, while it was inferior to HD-NOMA system. In terms of ergodic sum rate, the location of the MiBS exerted an influence on it and FD-NOMA revealed its preponderance. Furthermore, the sensing performance was limited by the transmit power of the MaBS, relay operation mode and reflection factor. Finally, we could confirm that the proposed OPA scheme is superior to other RPA schemes.

\numberwithin{equation}{section}
\section*{Appendix~A: Proof of Theorem 1} \label{Appendix:A}
\renewcommand{\theequation}{A.\arabic{equation}}
\setcounter{equation}{0}
Substituting (\ref{sinrsdf}) and (\ref{sinryrf}) into (\ref{oprdf}), $I_1$ can be rewritten as
\begin{align}\nonumber
  {I_1} &= \Pr \left( {\frac{{{a_f}{\rho _{SR}}{\gamma _c}}}{{{a_n}{\rho _{SR}}{\gamma _c} + {\rho _{RR}}\delta {\gamma _r} + {\rho _{LI}}\omega {\gamma _r} + 1}} < {\gamma _{{\rm{th}}_f}}} \right)  \Pr \left( {\frac{{{a_f}{\rho _{S{D_f}}}{\gamma _c}}}{{{a_n}{\rho _{S{D_f}}}{\gamma _c} + 1}} < {\gamma _{{\rm{th}}_f}}} \right) \\\nonumber
   &= \Pr \left( {{\rho _{SR}} < {\rho _{RR}}\delta {\gamma _r}{\theta _1} + {\rho _{LI}}\omega {\gamma _r}{\theta _1} + {\theta _1}} \right)\Pr \left( {{\rho _{S{D_f}}} < {\theta _1}} \right) \\\label{i1}
  & =\underbrace {\int_0^\infty  {\int_0^\infty  {\int_0^{{\theta _1}\left( {x\delta {\gamma _r} + y\omega {\gamma _r} + 1} \right)} {{f_{{\rho _{RR}}}}\left( x \right){f_{{\rho _{LI}}}}\left( y \right){f_{{\rho _{SR}}}}\left( z \right)dxdydz} } } }_{{I_{1,{\rm I}}}}  \underbrace {\int_0^{{\theta _1}} {{f_{{\rho _{S{D_f}}}}}\left( w \right)} dw}_{{I_{1,{\rm I}{\rm I}}}}.
\end{align}

By using PDFs and CDFs of ${{\rho _i}}$, ${I_{1,{\rm I}}}$ can be expressed as
\begin{equation}\label{i11}
  {I_{1,{\text{I}}}} = 1 - \frac{{{\beta _{SR}}}}{{2{\beta _{RR}}\left( {{\beta _{SR}} + {\beta _{LI}}\omega {\gamma _r}{\theta _1}} \right)}}{e^{ - \frac{{{\theta _1}}}{{{\beta _{SR}}}}}} \underbrace {\int_0^\infty  {\frac{1}{{\sqrt x }}{e^{ - \frac{{\sqrt x }}{{{\beta _{RR}}}} - \frac{{x\delta {\gamma _r}{\theta _1}}}{{{\beta _{SR}}}}}}} dx}_{I_{1,{\rm I}}^\Delta }.
\end{equation}

Denoting $t \triangleq \sqrt x $ and using \cite[Eq. (3.322.2)]{2007GradshteynI}, ${I_{1,{\rm I}}^\Delta }$ can be further expressed as
\begin{equation}\label{il3}
    I_{1,{\rm I}}^\Delta  = 2\int_0^\infty  {{e^{ - \frac{t}{{{\beta _{RR}}}} - \frac{{\delta {\gamma _r}{\theta _1}{t^2}}}{{{\beta _{SR}}}}}}} dt = \sqrt {\frac{{\pi {\beta _{SR}}}}{{\delta {\gamma _r}{\theta _1}}}} {e^{\frac{{{\beta _{SR}}}}{{4\beta _{RR}^2\delta {\gamma _r}{\theta _1}}}}}\left[ {1 - {\text{erfc}}\left( {\frac{1}{{2{\beta _{RR}}}}\sqrt {\frac{{{\beta _{SR}}}}{{\delta {\gamma _r}{\theta _1}}}} } \right)} \right].
\end{equation}

Substituting the PDF of $\rho_i$ into the second part of (\ref{i1}), after some straight forward mathematical manipulations, $I_1$ can be finally expressed as
\begin{align}\label{i14}
{I_1} \!=\! \left( \!{1 \!-\! {e^{ - \frac{{{\theta _1}}}{{{\beta _{S{D_f}}}}}}}} \!\right)\!\left\{\! {1 \!-\! \frac{{{\beta _{SR}}}}{{2{\beta _{RR}}\left( {{\beta _{SR}} \!+\! {\beta _{LI}}\omega {\gamma _r}{\theta _1}} \right)}}{e^{\frac{{{\beta _{SR}}}}{{4\beta _{RR}^2\delta {\gamma _r}{\theta _1}}} \!- \! \frac{{{\theta _1}}}{{{\beta _{SR}}}}}}\sqrt {\frac{{\pi {\beta _{SR}}}}{{\delta {\gamma _r}{\theta _1}}}} \left[ \!{1 \!-\! {\text{erfc}}\!\left(\! {\frac{1}{{2{\beta _{RR}}}}\sqrt {\frac{{{\beta _{SR}}}}{{\delta {\gamma _r}{\theta _1}}}} } \!\right)} \!\right]} \!\right\}.
\end{align}

Substituting (\ref{sinrsdf}), (\ref{sinryrf}), and (\ref{sinrff}) into (\ref{oprdf}), $I_2$ can be written as
\begin{equation}\label{i2}
  {I_2} = \Pr \left( {{\rho _{SR}} > {\rho _{RR}}\delta {\gamma _r}{\theta _1} + {\rho _{LI}}\omega {\gamma _r}{\theta _1} + {\theta _1}} \right) \Pr \left( {{\rho _{S{D_f}}} < {\theta _1}} \right)\Pr \left( {{\rho _{R{D_f}}} < {\varphi _1}} \right).
\end{equation}

Similar to $I_1$, after some mathematical manipulations, $I_2$ can be expressed as
\begin{align}\nonumber
{I_2}& = \left( {1 - {e^{ - \frac{{{\theta _1}}}{{{\beta _{S{D_f}}}}}}}} \right)\left( {1 - {e^{ - \frac{{{\varphi _1}}}{{{\beta _{R{D_f}}}}}}}} \right)\frac{{{\beta _{SR}}}}{{2{\beta _{RR}}\left( {{\beta _{SR}} + {\beta _{LI}}\omega {\gamma _r}{\theta _1}} \right)}} \hfill \\\label{i22}
   & \times {e^{\frac{{{\beta _{SR}}}}{{4\beta _{RR}^2\delta {\gamma _r}{\theta _1}}} - \frac{{{\theta _1}}}{{{\beta _{SR}}}}}}\sqrt {\frac{{\pi {\beta _{SR}}}}{{\delta {\gamma _r}{\theta _1}}}} \left[ {1 - {\text{erfc}}\left( {\frac{1}{{2{\beta _{RR}}}}\sqrt {\frac{{{\beta _{SR}}}}{{\delta {\gamma _r}{\theta _1}}}} } \right)} \right].
\end{align}

Substituting (\ref{i14}) and (\ref{i22}) into (\ref{oprdf}), after some straight-forward mathematical manipulations, (\ref{oprf1}) can be obtained.
\vspace{-4 mm}
\section*{Appendix~B: Proof of Corollary 1} 
\vspace{-2 mm}
\renewcommand{\theequation}{B.\arabic{equation}}
\setcounter{equation}{0}
In the high SNR regime, i.e., ${\gamma _c},{\gamma _r} \to \infty $, we obtain
\vspace{-1 mm}
\begin{align}\label{theta1}
\mathop {\lim }\limits_{{\gamma _c} \to \infty } {\theta _1} &= \mathop {\lim }\limits_{{\gamma _c} \to \infty } \frac{{{\gamma _{{\rm{th}}_f}}}}{{{a_f}{\gamma _c} - {a_n}{\gamma _c}{\gamma _{{\rm{th}}_f}}}} \approx 0,\\\label{phi1}
\mathop {\lim }\limits_{{\gamma _r} \to \infty } {\varphi _1} &= \mathop {\lim }\limits_{{\gamma _r} \to \infty } \frac{{{\gamma _{{\rm{th}}_f}}}}{{{a_f}{\gamma _r} - {a_n}{\gamma _r}{\gamma _{{\rm{th}}_f}}}} \approx 0,\\\label{thetax}
\mathop {\lim }\limits_{{\gamma _c},{\gamma _r} \to \infty } {\gamma _r}{\theta _1} &= \mathop {\lim }\limits_{{\gamma _r} \to \infty } \frac{{{\gamma _r}{\gamma _{{\rm{th}}_f}}}}{{{a_f}{\gamma _c} - {a_n}{\gamma _c}{\gamma _{{\rm{th}}_f}}}} \approx \theta _1^*.
\end{align}

Substituting (\ref{thetax}) into (\ref{i1}), $I_{1,{\rm I}}$ in high SNR regime can be expressed as
\begin{equation}
\label{ild}
I_{1,{\rm I}}^\infty  \!=\!\! \int_0^\infty \!\! {\int_0^\infty \!\! {\int_0^{x\delta \theta _1^* + y\omega \theta _1^*} \!\!{{f_{{\rho _{RR}}}}\left( x \right){f_{{\rho _{LI}}}}\left( y \right){f_{{\rho _{SR}}}}\left( z \right)dxdydz}}}.
\end{equation}
Similar to (\ref{i11}), after some mathematical manipulations, we can obtain the expression of $I_{1,{\rm I}}^\infty$.

When $x \to 0$, by using the approximation $1 - {e^x} \approx x$ and substituting (\ref{theta1}) into (\ref{i1}), $I_{1,{\rm {II}}}$ in high SNR regime is expressed as $I_{1,{\rm I}{\rm I}}^\infty  = {{{\theta _1}} \mathord{\left/
 {\vphantom {{{\theta _1}} {{\beta _{S{D_f}}}}}} \right.
 \kern-\nulldelimiterspace} {{\beta _{S{D_f}}}}}$.

Likewise, we can obtain $I_2^{\infty}$ by using the same method. Substituting $I_1^{\infty}$ and $I_2^{\infty}$ into (\ref{oprdf}), (\ref{asdf2}) can be obtained.
\section*{Appendix~C: Proof of Theorem 3} 
\renewcommand{\theequation}{C.\arabic{equation}}
\setcounter{equation}{0}
Defining ${W_1} \triangleq \min \left( {\gamma _{SR}^{{D_f}},\gamma _{R{D_f}}^{{D_f}},\gamma _{R{D_n}}^{{D_f}}} \right)$ and ${W_2}\triangleq \min \left( {\gamma _{S{D_f}}^{{D_f}},\gamma _{S{D_n}}^{{D_f}}} \right)$, then the achievable rate of $D_f$ can be rewritten as
\begin{align}
\label{B1}
{R_f} = \underbrace {{{\log }_2}\left( {1 + {W_1}} \right)}_{R_f^1} + \underbrace {{{\log }_2}\left( {1 + {W_2}} \right)}_{R_f^2}.
\end{align}

Substituting (\ref{sinryrf}), (\ref{sinrff}), and (\ref{sinrdnf}) into $W_1$, with the aid of probability theory, the CDF of $W_1$ can be expressed as
\vspace{-1 mm}
\begin{align}\nonumber
  {F_{{W_1}}}\left( {{w_1}} \right) &= 1 - \Pr \left( {{\rho _{SR}} \geqslant \frac{{\delta {\gamma _r}{\rho _{RR}}}}{{{g_c}\left( {{w_1}} \right)}} + \frac{{\omega {\gamma _r}{\rho _{LI}}}}{{{g_c}\left( {{w_1}} \right)}} + \frac{1}{{{g_c}\left( {{w_1}} \right)}},{\rho _{R{D_f}}} \geqslant {g_r}\left( {{w_1}} \right),{\rho _{R{D_n}}} \geqslant {g_r}\left( {{w_1}} \right)} \right) \hfill \\\nonumber
   &= 1 \!-\! \frac{{{\beta _{SR}}{g_c}\left( {{w_1}} \right)}}{{2{\beta _{RR}}\left( {{\beta _{SR}}{g_c}\left( {{w_1}} \right) \!+\! {\beta _{LI}}\omega {\gamma _r}} \right)}}\sqrt {\frac{{\pi {\beta _{SR}}{g_c}\left( {{w_1}} \right)}}{{\delta {\gamma _r}}}} {e^{\frac{{{\beta _{SR}}{g_c}\left( {{w_1}} \right)}}{{4\beta _{RR}^2\delta {\gamma _r}}} - \frac{1}{{{\beta _{SR}}{g_c}\left( {{w_1}} \right)}} - \frac{1}{{{\beta _{R{D_f}}}{g_r}\left( {{w_1}} \right)}} - \frac{1}{{{\beta _{R{D_n}}}{g_r}\left( {{w_1}} \right)}}}} \hfill \\\label{B2}
   &\times \left[ {1 - {\text{erfc}}\left( {\frac{1}{{2{\beta _{RR}}}}\sqrt {\frac{{{g_c}\left( {{w_1}} \right){\beta _{SR}}}}{{\delta {\gamma _r}}}} } \right)} \right].
\end{align}

Next, we can further observe that
\begin{align}\nonumber
  R_{\rm{ave}}^{f,1} &= \int_0^\infty  {{{\log }_2}\left( {1 + {w_1}} \right){f_{{W_1}}}\left( {{w_1}} \right)} d{w_1}\\\label{ef1}
   & = \frac{1}{{\ln 2}}\int_0^\infty  {\frac{{1 - {F_{{W_1}}}\left( {{w_1}} \right)}}{{1 + {w_1}}}} d{w_1}.
\end{align}

Substituting (\ref{B2}) into (\ref{ef1}), we can obtain the expression of $R_{\rm{ave}}^{f,1}$. The same procedure can be adopted for the proof of $R_{\rm{ave}}^{f,2}$. Then, we can obtain (\ref{rate2f}) after some mathematical manipulations.

Likewise, we can derive $R_{\rm{ave}}^{n,1}$ allowing the same approach as in (\ref{B2}).

For $R_{\rm{ave}}^{n,2}$, denoting $U \triangleq \gamma _{S{D_n}}^{{D_n}}$ first, the CDF of $U$ can be written as
\begin{equation}\label{cdfu}
  {F_U}\left( u \right) = \Pr \left( {{a_n}{\rho _{S{D_n}}}{\gamma _c} \leqslant u} \right)= 1 - {e^{ - \frac{u}{{{\beta _{S{D_n}}}{a_n}{\gamma _c}}}}}.
\end{equation}

Similar to (\ref{ef1}), we can further obtain
\begin{align}\nonumber
  R_{\rm{ave}}^{n,2}& = \frac{1}{{\ln 2}}\int_0^\infty  {\frac{{1 - {F_{{W_4}}}\left( u \right)}}{{1 + u}}} du \\\label{avn2}
    & = \frac{1}{{\ln 2}}\int_0^\infty  {\frac{1}{{1 + u}}{e^{ - \frac{u}{{{\beta _{S{D_n}}}{a_n}{\gamma _c}}}}}du}.
\end{align}

With the aid of \cite[Eq. (3.352.4)]{2007GradshteynI}, the exact expression of $R_{\rm{ave}}^{n,2}$ can be expressed as
\begin{equation}\label{raven2}
R_{\rm{ave}}^{n,2} =  - \frac{1}{{\ln 2}}{e^{\frac{1}{{{\beta _{S{D_n}}}{a_n}{\gamma _c}}}}}{\text{Ei}}\left( { - \frac{1}{{{\beta _{S{D_n}}}{a_n}{\gamma _c}}}} \right).
\end{equation}

After some straight forward mathematical manipulations, (\ref{rate2n}) can be obtained.
\section*{Appendix~D: Proof of Corollary 3} 
\renewcommand{\theequation}{D.\arabic{equation}}
\setcounter{equation}{0}
Substituting (\ref{sinryrf}), (\ref{sinrff}), and (\ref{sinrdnf}) into (\ref{r2f}), and using the inequality \cite{lix20}
\begin{equation}\label{ineq}
\mathbb{E}\left[ {{{\log }_2}\left( {1 + \frac{x}{y}} \right)} \right] \approx {\log _2}\left( {1 + \frac{{\mathbb{E}\left[ x \right]}}{{\mathbb{E}\left[ y \right]}}} \right),
\end{equation}
$R_{\rm{ave}}^{f,1}$ can be approximated as
\vspace{-1 mm}
\begin{align}\nonumber
  R_{\rm{ave}}^{ap,f1} &\approx \frac{1}{{\ln 2}}\ln \left[ {1 + \min \left( {\frac{{{a_f}\mathbb{E}\left[ {{\rho _{R{D_f}}}} \right]{\gamma _r}}}{{{a_n}\mathbb{E}\left[ {{\rho _{R{D_f}}}} \right]{\gamma _r} + 1}},\frac{{{a_f}\mathbb{E}\left[ {{\rho _{R{D_n}}}} \right]{\gamma _r}}}{{{a_n}\mathbb{E}\left[ {{\rho _{R{D_n}}}} \right]{\gamma _r} + 1}},} \right.} \right. \hfill \\\label{asavf1}
  &\;\;\;\;\;\;\;\;\;\;\;\;\;\;\;\;\;\;\;\;\;\;\;\;\;\left. {\left. {\frac{{{a_f}\mathbb{E}\left[ {{\rho _{SR}}} \right]{\gamma _c}}}{{{a_n}\mathbb{E}\left[ {{\rho _{SR}}} \right]{\gamma _c} + \mathbb{E}\left[ {{\rho _{RR}}} \right]\delta {\gamma _r} + \mathbb{E}\left[ {{\rho _{LI}}} \right]\omega {\gamma _r} + 1}}} \right)} \right],
\end{align}
where $\mathbb{E}\left[{\rho_{SR}}\right]$ is expressed as
\begin{equation}\label{esr}
  \mathbb{E}\left[ {{\rho _{SR}}} \right]= \int_0^\infty  {x{f_{{\rho _{SR}}}}\left( x \right)} dx  = {\beta _{SR}}.
\end{equation}

Then, we can obtain $\mathbb{E}\left[{\rho_{RD_f}}\right]$, $\mathbb{E}\left[{\rho_{RD_n}}\right]$, $\mathbb{E}\left[{\rho_{RR}}\right]$, and $\mathbb{E}\left[{\rho_{LI}}\right]$ with the same method. $R_{\rm{ave}}^{ap,f1}$ can be ontained after some mathematical calculations.

Similarly, $R_{\rm{ave}}^{ap,f2}$, $R_{ave}^{ap,n1}$ and $R_{ave}^{ap,n2}$ can be derived. Then, the approximate ergodic sum rate of the IoT devices can be further obtained.


\end{document}